\documentclass[11pt,letter]{article}
\usepackage{float}
\usepackage{hyperref}
\usepackage{lineno}

\usepackage{epsfig}
\usepackage{amssymb}
\usepackage{lscape,graphicx}
\usepackage{array}
\usepackage{setspace}
\usepackage{fullpage}
\usepackage{centernot}
\usepackage{bbm}
\usepackage{chemarr}
\usepackage[margin=1in]{geometry}
\usepackage{authblk}

\makeatletter
\renewcommand{\@biblabel}[1]{\quad#1.}
\makeatother

      \newcommand {\mm}[1] {\ifmmode{#1}\else{\mbox{\(#1\)}}\fi}
      \newcommand{\real} {\mm{{\mathbb R}}}

      \newcolumntype{^}{>{\currentrowstyle}}

\newcommand{\utwi}[1]{\mbox{\boldmath $ #1$}}

\newcommand{\bp}{{\utwi{p}}}

\newcommand{\bs}{{\utwi{s}}}

\newcommand{\bx}{{\utwi{x}}}
\newcommand{\by}{{\utwi{y}}}

\newcommand{\bA}{{\utwi{A}}}
\newcommand{\bB}{{\utwi{B}}}

\newcommand{\bP}{{\utwi{P}}}
\newcommand{\bQ}{{\utwi{Q}}}
\newcommand{\bR}{{\utwi{R}}}

\newcommand{\bpi}{{\utwi{\pi}}}

\newcommand{\diag} {\textrm{diag}}

\DeclareMathOperator{\Err}{Err}

\DeclareMathAlphabet{\mathpzc}{OT1}{pzc}{m}{it}

\newtheorem{theorem}{Theorem}
\newtheorem{lemma}{Lemma}

\newenvironment{proof}[1][Proof]{\begin{trivlist}
\item[\hskip \labelsep {\bfseries #1}]}{\end{trivlist}}

\title{State space truncation with quantified errors for accurate solutions to discrete Chemical Master Equation}
\author[1,2]{Youfang Cao\thanks{ycao@lanl.gov}}
\author[1]{Anna Terebus}
\author[1,3]{Jie Liang\thanks{jliang@uic.edu}}
\affil[1]{Department of Bioengineering, University of Illinois at Chicago, Chicago IL}
\affil[2]{Current address: Theoretical Biology and Biophysics (T-6), Center for Nonlinear Studies (CNLS), Los Alamos National Laboratory, Los Alamos, NM}
\affil[3]{Corresponding author}

\begin{document}
\date{}
\maketitle

\begin{abstract}
The discrete chemical master equation (dCME) provides a general
framework for studying stochasticity in mesoscopic reaction networks.
Since its direct solution rapidly becomes intractable due to the
increasing size of the state space, truncation of the state space is
necessary for solving most dCMEs. It is therefore important to assess
the consequences of state space truncations so errors can be
quantified and minimized.  Here we describe a novel method for state
space truncation.  By partitioning a reaction network into multiple
molecular equivalence groups (MEG), we truncate the state space by
limiting the total molecular copy numbers in each MEG.  We further
describe a theoretical framework for analysis of the truncation error
in the steady state probability landscape using reflecting boundaries.
By aggregating the state space based on the usage of a MEG and
constructing an aggregated Markov process, we show that the truncation
error of a MEG can be asymptotically bounded by the probability of
states on the reflecting boundary of the MEG.  Furthermore, truncating
states of an arbitrary MEG will not undermine the estimated error of
truncating any other MEGs.  We then provide an overall error estimate
for networks with multiple MEGs.  To rapidly determine the appropriate
size of an arbitrary MEG, we also introduce an {\it a priori}\/ method to
estimate the upper bound of its truncation error. This {\it a
  priori}\/ estimate can be rapidly computed from reaction rates of
the network, without the need of costly trial solutions of the dCME.
As examples, we show results of applying our methods to the four
stochastic networks of 1) the birth and death model, 2) the single
gene expression model, 3) the genetic toggle switch model, and 4) the
phage lambda bistable epigenetic switch model.  We demonstrate how
truncation errors and steady state probability landscapes can be
computed using different sizes of the MEG(s) and how the results
validate out theories.  Overall, the novel state space truncation and
error analysis methods developed here can be used to ensure accurate
direct solutions to the dCME for a large number of stochastic
networks.

\end{abstract}

\section*{Introduction}

Biochemical reaction networks are intrinsically 
stochastic~\cite{Stewart2012,Qian2012}. 
Deterministic models based on chemical mass action kinetics cannot
capture the stochastic nature of these
networks~\cite{McAdams1999,Wilkinson2009,Cao2010}.  Instead, the
discrete Chemical Master Equation (dCME) that describes the
probabilistic reaction jumps between discrete states 
provides a general framework for fully characterizing
mesoscopic stochastic processes in a well mixed
system~\cite{Gillespie_JPC77,Gillespie-PhysicaA-1992,vanKampen2007,Beard2008,Gillespie2009-jcp}.
The steady state and time-evolving probability landscapes over discrete states
governed by the dCME  provide detailed information of these
dynamic stochastic processes.  However, the dCME cannot be solved
analytically, except for a few very simple
cases~\cite{Darvey-1966-jcp,McQuarrie-1967-JAppProb,Taylor1998,Laurenzi-2000-jcp,Vellela2007}.

 The dCME can be approximated using the Fokker-Planck equation (FPE)
 and the chemical Langevin equation (CLE). These approximations are
 not applicable when copy numbers are small~\cite{Gillespie_JCP2000},
 as relatively large copy numbers of molecules are required for
 accurate
 approximation~\cite{VanKampen-1961,Gillespie_JCP2000,Gillespie2002,Haseltine2002,Gardiner2004-book}.
 Recent studies provided assessment of errors in these approximations
 for several reaction
 networks~\cite{Grima-2011-jcp,Grima-2013-bcmgenomics}, as well as
 numerical demonstration in which the CLE of a 13-node lysogeny-lysis
 decision network of phage-lambda was found to fail to converge to the
 correct steady state probability landscape (see appendix of
 ref~\cite{Cao2010}). However, consequences of such approximations
 involving many molecular species and with complex reaction schemes
 are generally not known.

  A widely used approach to
  study stochasticity is that of stochastic simulation algorithm (SSA)
  It generates reaction trajectories following the underlying
  dCME~\cite{Gillespie_JPC77}, and the stochastic properties of the
  network can then be inferred through analysis of a large number of
  simulation trajectories. However, convergence of such simulations is
  difficult to determine, and the errors in the sampled steady state
  probability landscape are unknown.

Directly solving the dCME offers another attractive approach.  By
computing the probability landscape of a stochastic network
numerically, its properties, such as those involving rare events, can
be studied accurately in details.  The finite state projection (FSP)
method is among several methods that have been developed to solve dCME
directly~\cite{Munsky2006,CaoBMCSB08,MacNamara2008a,MacNamara2008b,Cao2010,Wolf2010,Jahnke2011}.
The FSP is based on a truncated projection of the state space and
uses numerical techniques to compute the time-evolving probability
landscapes, which are solutions to the
dCME~\cite{Sidje1998,Munsky2006,Munsky2007}.  Although the error due
to state space truncation can be calculated for the time-evolving
probability landscape~\cite{Munsky2006}, the use of an absorbing
boundary, to which all truncated states are projected, will lead to
the accumulation of errors as time proceeds, and 
eventually trap all probability mass.  The FSP method was designed to
study transient behavior of stochastic networks, and is not well
suited to study the long-term behavior and the steady state
probability landscape of a network.

A bottleneck problem for solving the dCME directly is to have an
efficient and adequate account of the discrete state space.  As the
copy number of each of the $n$ molecular species takes an integer
value, conventional hypercube-based methods of state enumeration
incorporate all vertices in a $n$-dimensional hypercube non-negative
integer lattice, which has an overall size of $ O
(\prod_{i=1}^{n}{b_i})$, where $b_i$ is the maximally allowed copy
number of molecular species $i$.  State enumeration rapidly becomes
intractable, both in storage and in computing time. This makes the
direct solution of the dCME impossible for many realistic problems.
To address this issue, the finite buffer discrete CME (fb-dCME) method
was developed for efficient enumeration of the state
space~\cite{CaoBMCSB08}.  This algorithm is provably optimal in both
memory usage and in time required for enumeration. It introduces a
buffer queue with a fixed number of molecular tokens to keep track of
the remaining number of states that can be enumerated.  States with
depleted buffer do not absorb probability mass but reflect them to
states already enumerated, with the overall probability mass
conserved.  Further, instead of including every states in a hypercube,
it examines only states that can be reached from a given initial
state.  It can be used to compute the exact steady state and
time-evolving probability landscape of a closed network, or an open
network when the net gain in newly synthesized molecules does not
exceed the predefined finite buffer capacity.

State-space truncation eventually occurs in all methods that directly
solve the dCME.  For example, it occurs in open systems when no new
states can be enumerated, therefore synthesis reaction cannot proceed.
However, it is unclear how accurate the probability landscape computed
using a truncated state space is. Furthermore, it is unclear how to
minimize truncation errors, thus limiting the scope of applications of
direct methods such as the fb-dCME method.

In this study, we develop a new method for state space truncation and
provide a general theoretical framework for characterizing the error
due to state space truncation.  We start by partitioning the molecular
species in a reaction network into a number of {\it molecular
  equivalent groups}\/ (MEG) according to their chemical
compositions. The state space is then truncated by limiting the
maximum copy number of each MEG instead of individual molecular
species.  States with exactly the maximum copy number of a MEG form
the reflecting boundary of the state space. We further discuss
networks with a single reflecting boundary in the truncated state
space.  We then show that the total probability of the boundary states
can be used as an upper bound of the truncation error in computed
steady state probability landscape. This is then generalized to
networks with an arbitrary number of reflecting boundaries.  We
further develop an {\it a priori}\/ method derived from stochastic
ordering for rapid estimation of the truncation errors of the steady
state probability landscape for a given truncated state space.  The
required maximum copy number of each MEG for a pre-defined error
tolerance can also be determined without computing costly trial
solutions to the dCME.  Overall, the method of state space truncation
and the upper bounds of truncation errors established in this study
enables accurate quantification of errors in numerical solutions of
the dCME, and can help to design strategies so probability landscapes
with small and controlled errors can be computed for a large class of
biological problems which are previously infeasible.

This paper is organized as follows.  We first review basic concepts of
the discrete chemical master equation and issues associated with the
finite discrete state space. We then describe how to partition a
reaction network into molecular equivalent groups and how to truncate
the discrete state space.  We further discuss truncation errors of the
steady state probability landscape and how to construct upper bounds
of the truncation errors.  This is followed by detailed studies of the
single gene expression system and the genetic toggle switch system.
We examine the \textit{a priori} estimated error bound, the computed
error, and the true error for different state truncations.  We end
with discussions and conclusions.

\section*{Methods}

\subsection*{Theoretical Framework}

\subsubsection*{Reaction Network, State Space and Probability Landscape}

In a well-mixed biochemical system with constant volume and
temperature, there are $n$ molecular species, denoted as ${\mathcal X}
= \{ X_1 , X_2, \cdots, X_n \}$, and $m$ reactions, denoted as
${\mathcal R} = \{ R_1, R_2, \cdots, R_m \}$. Each reaction $R_k$ has
an intrinsic reaction rate constant $r_k$.  The microstate of the
system at time $ t $ is given by the non-negative integer column
vector $\bx(t) \in \mathbb{Z}_+^n$ of copy numbers of each molecular
species: $\bx (t) = (x_1(t), x_2(t), \cdots, x_ n(t) )^T$, where
$x_i(t)$ is the copy number of molecular species $X_i$ at time $t$. An
arbitrary reaction $R_k$ with intrinsic rate $r_k$ takes the general
form of
$$
c_{1k}X_1 + c_{2k}X_2 + \cdots + c_{nk}X_n
\overset{r_k}{\rightarrow} c'_{1k}X_1 + c'_{2k}X_2 + \cdots +
c'_{nk}X_n,
$$ which brings the system from a microstate $\bx_j$ to $\bx_i$.  
The difference between $\bx_i$ and $\bx_j$ is the
stoichiometry vector $\bs_k$ of reaction $R_k$: $ \bs_k = \bx_i -
\bx_j = (s_{1k}, s_{2k},\cdots, s_{nk})^T =( c'_{1k}-c_{1k},\,
c'_{2k}-c_{2k},\, \cdots,\, c'_{nk}-c_{nk})^T \in { \mathbb Z}^n.$ 
The rate $A_k(\bx_i, \bx_j)$ of reaction $R_k$ that brings the microstate
from $\bx_j$ to $\bx_i$ is determined by $r_k$ and the combination
number of relevant reactants in the current microstate $\bx_j$:
$$
A_k(\bx_i, \bx_j) = A_k (\bx_j) = r_k \prod_{l=1}^n
\binom{x_l}{c_{lk}},
$$ 
assuming the convention $\binom{0}{0} = 1$.

All possible microstates that a system can visit from a given
initial condition form the state space:
${\Omega} = \{\bx(t) | \bx(0), \, t \in (0, \, \infty)\}.$ We
denote the probability of each microstate at time $t$ as $p(\bx(t))$,
and the probability distribution at time $t$ over the full state
space as ${\bp}(t) = \{(p(\bx(t)) | \bx(t) \in \Omega) \}.$ We
also call ${\bp}(t)$ the {\it probability landscape\/} of the
network~\cite{Cao2010}.

\subsubsection*{Discrete Chemical Master Equation}

The discrete chemical master equation (dCME) can be written as a set
of linear ordinary differential equations describing the change in probability
of each discrete state over time:
\begin{equation}
\frac{d p(\bx, t)}{dt} = \sum_{\bx',\, \bx' \neq \bx  } 
[ A(\bx, \bx') p (\bx', t) - A(\bx', \bx) p (\bx, t)],
\label{eqn:dcme1}
\end{equation}
Note that $p(\bx, t)$ is
continuous in time, but is discrete over the state space. 
In matrix form, the dCME can be written as: 
\begin{equation}
\frac{d \bp(t)}{dt} = \bA \bp(t), 
\label{eqn:dcme2}
\end{equation}
where $\bA  \in 
\real^{|\Omega| \times |\Omega|}$
is the transition rate matrix formed by the collection of
all $A(\bx_i, \bx_j)$, which describes the overall 
reaction rate from state $\bx_j$ to state $\bx_i$: 
\begin{equation}
 A(\bx_i,\bx_j) = \left\{ 
  \begin{array}{l l}
    \sum_{k=1}^m A_k(\bx_i, \bx_j) & \quad \text{if $\bx_i \neq \bx_j$ and $\bx_j \stackrel{R_k}{\longrightarrow} \bx_i$}, \\
    - \sum_{\substack{\bx' \in \Omega, \\ \bx' \neq \bx_j}} A(\bx', \bx_j) & \quad \text{if $\bx_i = \bx_j$}, \\
		0 & \quad \text{otherwise}. \\
  \end{array} \right.
\label{eqn:matA}
\end{equation}

\subsubsection*{Molecular Equivalent Groups and Independent Birth-Death Processes}

In an open reaction network, synthesis reactions are the only ones
that generate new molecules and increase the total mass of the system.
Degradation reactions are the only ones that destroy molecules and
remove mass from the system.  The net copy numbers of various
molecular species in an open network gives its total mass.  For a
given microstate, the mass for each molecular species is defined.  The
total mass in a network can increase to infinity if synthesis
reactions persist.  The truncation of the infinite state space of such
an open network, which is inevitable due to the limited computing
capacity, can lead to errors in computing the probability landscapes
of a dCME.

Here we introduce the concept of \textit{Molecular Equivalence Groups}
(MEG), which will be useful for state space truncation.  Specifically,
molecular species $X_i$ and $X_j$ belong to the same MEG if $X_i$ can
be transformed into $X_j$ or $X_j$ can be transformed into $X_i$
through one or more \textbf{mass-balanced reactions}.  A stochastic
network can have one or more  Molecular Equivalent Groups.  The total
mass of a Molecular Equivalent Group for a specific microstate is
defined as the total copy number of the most elementary equivalent
molecular species in the Molecular Equivalent Group.

\begin{equation}
	\begin{aligned}
	A &\rightarrow B; \quad 
	2A &\rightarrow C; \quad 
	B + C &\rightarrow D; \quad 
	2X &\rightarrow Y; \quad 
	Y &\rightarrow Z; \quad  
	\end{aligned}
	\label{eqn:emc}
\end{equation}

For example, the reaction network shown in Eqn.~(\ref{eqn:emc}) has
two Molecular Equivalent Groups, {\it i.e.},  $MEG_1 = \{ A, B, C, D \}$
and $MEG_2 = \{ X, Y, Z \}$.  The most elementary molecular
species in $MEG_1$ and $MEG_2$ are $A$ and $X$, respectively.  For
any specific microstate of the network $\bx = \{ a, b, c, d, x, y,
z \}$, the total net copy number of the Molecular Equivalent Group
$MEG_1$ is calculated as $n_{MEG_1} (\bx) = a + b
+ 2c + 3d$, and the total net copy number of $MEG_2$ can
be calculated as $n_{MEG_2} (\bx) = x + 2y +2z$, where
the $a$, $b$, $c$, $d$, $x$, $y$, and $z$ are copy numbers of
corresponding molecular species.

We are interested in MEGs containing synthesis and degradation
reactions.  The set of reactions associated with such an open MEG is
called an \textit{independent Birth-Death process} (iBD).  Reactions
in an iBD can increase or decrease the total net copy number of
molecules in the associated MEG.

\subsubsection*{State Space Truncation by Molecular Equivalent Group}

Here we introduce a novel state truncation method.  Instead of
truncating the state space by specifying a maximum allowed copy number
$B$ for each molecular species, we specify a maximum allowed molecular
copy number $B$ for the $j$-th MEG.  Assume the $j$-th MEG contains
$n_j$ distinct molecular species, and conservatively ignore the
effects of stoichiometry, the number of all possible states for the
$j$-th MEG is then that of the volume of an $n_j$-dimensional
orthogonal corner simplex, with $B$ the length of all edges with the
origin as a vertex.  The number of integer lattice nodes in this
$n_j$-dimensional simplex gives the precise number of states of the
$j$-th MEG, which is in turn exactly given by the multiset number $\binom{B +
  n_j}{n_j}$.  The size of the state space is therefore much smaller
than the size of the state space $B^{n_j}$ that would be
generated by the hypercube method, with a reduction factor of roughly
$n_j !$ factorial.  Note that under the constraint of mass
conservation, each molecular species in this MEG can still have a
maximum of $B$ copies of molecules.

We further conservatively assume that different MEGs are independent,
and each can have maximally $B$ copies of molecules.  The size of the
overall truncated state space is then $O(\prod_j \binom{B +
  n_j}{n_j})$.  This is much smaller than the $n$-dimensional
hypercube, which has an overall size of $O( \prod_j B^{n_j}) =
O(B^n)$, with $n$ the total number of molecular species in the
network.  Overall, the size of state space generated by MEG truncation
can be dramatically smaller than that generated using the hypercube
method.

\subsubsection*{State Space Aggregation According to the Net Copy
Number in Molecular Equivalent Group}

We first consider the stochastic network with only one Molecular
Equivalent Group.  We truncate the state space by fixing the maximum amount of
total mass in the network.  We are interested in estimating the errors
due to such a state truncation.  To do so, we first factor states in
the original state space $\Omega^{(\infty)}$ of infinite size
according to the total net copy number of the MEG in each state.  The
infinite state space $\Omega^{(\infty)}$ can be partitioned into
disjoint groups of subsets $\tilde{\Omega}^{(\infty)} \equiv \{
\mathcal{G}_0, \mathcal{G}_1, \cdots, \mathcal{G}_N, \cdots \}$,
where states in each aggregated subset $\mathcal{G}_s$ have exactly
the same $s$ total copies of equivalent elementary molecular species
of the MEG.  The total steady state probability
$\tilde{\pi}_s^{(\infty)}$ on microstates in each group
$\mathcal{G}_s$ can then be written as:
\begin{equation}
\tilde{\pi}_s^{(\infty)} \equiv
\sum_{\bx \in \mathcal{G}_s}
\pi^{(\infty)}(\bx) =
\sum_{\bx \in \mathcal{G}_s}
p^{(\infty)}(\bx,\, t=\infty). 
\label{eqn:computedError}
\end{equation}

Based on the state space partition $\tilde{\Omega}^{(\infty)}$, we can re-construct
a transition rate matrix $\tilde{\bA}$, which is a permutation of the original dCME
matrix $\bA$ in Eqn.~(\ref{eqn:dcme2}):
\begin{equation}
\tilde{\bA} = \left( {\begin{array}{*{20}c}
   {\bA_{i,j}} 
\end{array}} \right), \quad 0 \leq i,j \leq \infty
\label{eqn:Aaggreg1}
\end{equation}
where each block sub-matrix $\bA_{i,\,j}$ includes all transitions
from states in group $\mathcal{G}_j$ to states in $\mathcal{G}_i$.

In continuous time Markov model of mesoscopic systems, reactions occur
instantaneously, and the synthesis and degradation reactions always
generate or destroy one molecule at a time.  This also applies to
oligomers, which are assumed to form only upon association of monomers
already synthesized, and dissociate into monomers first before full
degradation.  The re-constructed matrix $\tilde{\bA}$ is thus a tri-diagonal
block matrix, \textit{i.e.}, $\bA_{i,\,j}$ is all 0s if $|i - j| > 1$.
Moreover, synthesis reactions always appear as lower blocks
$\bA_{i+1,\,i}$, and degradation reactions always as upper blocks
$\bA_{i,\,i+1}$.  Diagonal blocks $\bA_{i,\,i}$ contains all coupling
reactions that do not alter the net number of synthesized
molecules. Note that every $\bA_{i+1,\,i}$ block and $\bA_{i,\,i+1}$
block only includes synthesis and degradation reactions associated
with the current MEG.  For analysis of networks with multiple MEGs, we
assume at this time there is no limit on the total mass of other MEGs,
therefore the state space is not truncated on these MEGs.  These other
MEGs do not alter the total net copy number of molecular species in the
current MEG.

Note that the assumption of the stoichiometric coefficient of 1 for
synthesis and degradation is only for constructing the
proofs of the theorems.  In computation, there is no condition on the
stoichiometry of any reaction, and our method is general and can be
applied to any reaction network.

We can obtain the steady state probability $\tilde{\pi}_s^{(\infty)}$ 
on aggregated states without solving the dCME. 
It is tempting to lump all microstates in each group $\mathcal{G}_j$
into one state and replace the original $|\Omega^{(\infty)}| \times |\Omega^{(\infty)}|$
rate matrix $\tilde{\bA}$ with an aggregated matrix 
to study the dynamic changes of the probability landscape
on this aggregated state space.  However, stringent requirements
must be satisfied for such lumped states to follow a Markov
process~\cite{Tian2006,Truffet1997,Buchholz1994,Stewart1994,Vantilborgh1985,Kemeny1976}. 
Specifically, a transition rate matrix $\bA$ for a continuous Markov process is
lumpable with respect to a partition $\tilde{\Omega}^{(\infty)}$ if and only 
if for all pairs of $\mathcal{G}_s, \, \mathcal{G}_t \in \tilde{\Omega}^{(\infty)}$, 
the condition
\begin{equation}
\sum_{\bx_k \in \mathcal{G}_t} A_{ik} = \sum_{\bx_k \in \mathcal{G}_t} A_{jk}
\label{eqn:lumpcond}
\end{equation}
holds for all $\bx_i, \, \bx_j \in \mathcal{G}_s$~\cite{Tian2006}.  In
other words, every state in $\mathcal{G}_s$ must have the same total
transition rate to group $\mathcal{G}_t$, and this must be true for
all $\mathcal{G}_s$ and $\mathcal{G}_t$~\cite{Tian2006}.

While $\tilde{\bA}$ does not satisfy this strong condition in general,
we can instead construct a lumped transition matrix $\bB$, which is
associated with the aggregated state space derived from the partition
$\tilde{\Omega}^{(\infty)}$, such that the aggregated steady state
probability distribution on the partition $\tilde{\Omega}^{(\infty)}$
computed from the lumped matrix $\bB$ is equal to that derived from
the steady state distribution computed from the original matrix $\bA$.
That is, steady state probabilities on partitioned groups in
$\tilde{\Omega}^{(\infty)}$ are identical using either $\bB$ or the
original $\bA$.

Assume the steady state probability distribution $\tilde{\bpi}(\bx)$
over the partitioned state space $\tilde{\Omega}^{(\infty)}$ is known,
the aggregated synthesis rate $\alpha^{(\infty)}_i$ for the group
$\mathcal{G}_i$ and the aggregated degradation rate
$\beta^{(\infty)}_{i+1}$ for the group $\mathcal{G}_{i+1}$ at the
steady state are two constants (Fig~\ref{fig:bfbd}) defined as
\begin{equation}
\alpha^{(\infty)}_i = \left( \mathbbm{1}^T \bA_{i+1,i} \right) \cdot \frac{
\tilde{\bpi}(\mathcal{G}_{i})}{\mathbbm{1}^T \tilde{\bpi}(\mathcal{G}_{i})} \quad \text{and} \quad
\beta^{(\infty)}_{i+1} = \left( \mathbbm{1}^T \bA_{i,i+1} \right)
\cdot \frac{\tilde{\bpi}(\mathcal{G}_{i+1})}{\mathbbm{1}^T \tilde{\bpi}(\mathcal{G}_{i+1})},
\label{eqn:abdef}
\end{equation}
where $\tilde{\bpi}(\mathcal{G}_{i})$ and
$\tilde{\bpi}(\mathcal{G}_{i+1})$ are the steady state probability
vector over microstates in the lumped states $\mathcal{G}_{i}$ and
$\mathcal{G}_{i+1}$, respectively.  The term $\mathbbm{1}^T
\bA_{i+1,i}$ is the row vector of column-summed rates from
$\bA_{i+1,i}$ for microstates in $\mathcal{G}_i$, and $\frac{
  \tilde{\bpi}(\mathcal{G}_{i})}{\mathbbm{1}^T
  \tilde{\bpi}(\mathcal{G}_{i})}$ is the steady state probability
vector $\tilde{\bpi}(\mathcal{G}_{i})$ over microstates in
$\mathcal{G}_i$ normalized by the total steady state probability on
$\mathcal{G}_i$.  Similarly, $\mathbbm{1}^T \bA_{i,i+1}$ is the row
vector of column-summed rates from $\bA_{i,i+1}$ for microstates in
$\mathcal{G}_{i+1}$, and $\frac{
  \tilde{\bpi}(\mathcal{G}_{i+1})}{\mathbbm{1}^T
  \tilde{\bpi}(\mathcal{G}_{i+1})}$ is the steady state probability
vector $\tilde{\bpi}(\mathcal{G}_{i+1})$ over microstates in
$\mathcal{G}_{i+1}$ normalized by the total steady state probability
on $\mathcal{G}_{i+1}$.  We can construct an aggregated transition
rate matrix $\bB$ from $\tilde{\bA}$ based on the following Lemma:
\begin{lemma}
\label{lm:rma}
(Rate Matrix Aggregation.) If an Molecular Equivalence Group has no
limit on the total copy number, it generates an infinite state space
$\Omega^{(\infty)}$ and the rate matrix $\bA$ is of infinite
dimension.  For any homogeneous continuous-time Markov process with
such a rate matrix $\bA$, an aggregated continuous-time Markov process
with an infinite rate matrix $\bB^{(\infty)}$ can be constructed on
the partition $\tilde{\Omega}^{(\infty)} = \{ \mathcal{G}_0,
\mathcal{G}_1, \cdots, \mathcal{G}_N, \cdots \}$ with respect to the
total net copy number of molecules in the network, such that it gives
the same steady state probability distribution for each partitioned
group $\{\mathcal{G}_s\}$ as that given by the original matrix $\bA$,
\textit{i.e.}, $\pi(\mathcal{G}_s) = \sum_{\bx \in \mathcal{G}_s}
\bpi(\bx)$ for all $s = 0, 1, \cdots$, where $\bpi(\Omega^{(\infty)})$
is the steady state probability distribution associated with $\bA$.
Specifically, the infinite transition rate matrix $\bB^{(\infty)}$ can
be constructed as a tridiagonal matrix: 
\begin{equation}
\bB = \left( {\begin{array}{c}
   \boldsymbol\alpha^{(\infty)}, \boldsymbol\gamma^{(\infty)}, \boldsymbol\beta^{(\infty)}
\end{array}} \right),
\label{eqn:bdmatinf}
\end{equation}
with the lower off-diagonal vector $\boldsymbol\alpha^{(\infty)} = (\alpha^{(\infty)}_i)$,
the upper off-diagonal vector $\boldsymbol\beta^{(\infty)} = (\beta^{(\infty)}_{i+1})$,
and the diagonal vector $\boldsymbol\gamma^{(\infty)} = (\gamma^{(\infty)}_i) = (-\alpha^{(\infty)}_i - \beta^{(\infty)}_{i}), \, i = 0,\cdots,\infty$.
This is equivalent to transforming the corresponding infinite transition
rate matrix $\tilde{\bA}$ in Eqn.~(\ref{eqn:Aaggreg1}) 
into $\bB^{(\infty)}$ by substituting
each block sub-matrix $\bA_{i+1,\, i}$ of synthesis reactions with the
corresponding aggregated synthesis rate $\alpha^{(\infty)}_i$, and each
block $\bA_{i,\, i+1}$ of degradation reactions with the aggregated
degradation rate $\beta^{(\infty)}_{i+1}$, respectively, with 
$\alpha^{(\infty)}_i$ and $\beta^{(\infty)}_{i+1}$ 
defined in Eqn.~(\ref{eqn:abdef}). 
\end{lemma}

Proof can be found in the Appendix.

\subsubsection*{Analytical Solution of Steady State Probability of Aggregated States}

The system associated with the aggregated rate matrix $\bB$ can be
viewed as a birth-death process controlled by a pair of 
``synthesis'' and ``degradation'' transitions between aggregated states 
associated with different net copy number of the MEG. 
It takes the form: 
\begin{equation}
\emptyset \overset{\alpha^{(\infty)}_i}{\underset{\beta^{(\infty)}_{i+1}}{\rightleftharpoons}} \mathsf{E}, 
\label{eqn:bfbd}
\end{equation}
where $\mathsf{E}$ represents the elementary molecular species in the MEG, 
with its copy number  the total net copy number of the MEG. 
The rates $\alpha_i^{(\infty)}$ and $\beta_{i+1}^{(\infty)}$ are the aggregated 
``synthesis'' and ``degradation'' rates for this MEG. 
The aggregated state space and transitions between them are illustrated in
Fig.~\ref{fig:bfbd}.  The steady state probability distribution over the 
aggregated states are governed by $\bB \tilde{\bpi}^{(\infty)} = \utwi{0}$.

The aggregated rates $\alpha^{(\infty)}_i$ and
$\beta^{(\infty)}_{i+1}$ in $\bB$ are from summations of all entries
in the non-negative block matrices $\bA_{i+1,\, i}$ and $\bA_{i,\,
  i+1}$.  As long as there is one or more microstates in $\bA_{i+1,\,
  i}$ or $\bA_{i,\, i+1}$ with non-zero copies of reactants,
$\alpha^{(\infty)}_i$ or $\beta^{(\infty)}_i$ will be non-zero.  We next
examine the most general case when $\alpha^{(\infty)}_i \ne 0$ and
$\beta^{(\infty)}_{i+1} \ne 0$ for all $i = 0, \, 1,\, \cdots$.  We
simplify our notation and use $\tilde{\pi}_i^{(\infty)}$ for
$\tilde{\pi}^{(\infty)}(\mathcal{G}_i)$.  Following the well-known results
on analytical solution of the steady state distribution of the
birth-death processes~\cite{Taylor1998,Vellela2007}, the steady state solution for
$\tilde{\pi}_i^{(\infty)}$ and $\tilde{\pi}_0^{(\infty)}$ can be
written  as:
\begin{equation}
\tilde{\pi}_i^{(\infty)} = 
  \prod\limits_{k = 0}^{i-1}
    \frac{\alpha^{(\infty)}_{k}}{\beta^{(\infty)}_{k+1}} \tilde{\pi}_0^{(\infty)}, 
\label{eqn:pinii}
\end{equation}
and
\begin{equation}
\tilde{\pi}_0^{(\infty)} = 
\frac{1}
{1 + {\sum\limits_{j = 1}^{\infty} {\prod\limits_{k = 0}^{j-1}
     {\frac{\alpha^{(\infty)}_{k}}{\beta^{(\infty)}_{k+1}}}}
    }
}, 
\label{eqn:pin0}
\end{equation}
Therefore, the steady state probability $\tilde{\pi}_{i}^{(\infty)}$ 
of an arbitrary group $\mathcal{G}_{i}$ can be written as:
\begin{equation}
\tilde{\pi}_i^{(\infty)} = \frac{\prod\limits_{k = 0}^{i-1}
    \frac{\alpha^{(\infty)}_{k}}{\beta^{(\infty)}_{k+1}}}
{1 + {\sum\limits_{j = 1}^{\infty} {\prod\limits_{k = 0}^{j-1}
     {\frac{\alpha^{(\infty)}_{k}}{\beta^{(\infty)}_{k+1}}}} 
    }
}. 
\label{eqn:pininf}
\end{equation}
Once $\alpha^{(\infty)}_{k}$ and $\beta^{(\infty)}_{k+1}$ are known,
the total probability of any aggregated state $\mathcal{G}_i$ at the
steady state can be easily computed. We will introduce a method in
later sections for easy {\it a priori\/} calculation of error
estimates based on Eqn.~(\ref{eqn:pininf}) and values of
$\alpha^{(\infty)}_{k}$ and $\beta^{(\infty)}_{k+1}$, which are
directly obtained from reaction rate constants of the network model,
without the need of solving the dCME.

\subsubsection*{Truncation Error Is Bounded  Asymptotically by
Probability of Boundary States }
\label{sec:convergence}

When the maximum total net copy number of the MEG is limited to $N$, states 
with a total net copy number larger than $N$ will not be included, resulting in 
a truncated state space $\Omega^{(N)}$. Those microstates with exactly $N$ total net 
copies of molecules in the network are the \textit{boundary states}, 
because neighboring states with one additional molecule are truncated.  
The true error for the steady state $\Err^{(N)}$ due to truncating states
beyond those with $N$ net copies of molecules 
is the summation of true probabilities over microstates that have been truncated 
from the original infinite state space: 
\begin{equation}
\Err^{(N)} = \sum_{\bx \in \Omega^{(\infty)}, \, \bx \notin \Omega^{(N)}} \pi^{(\infty)}(\bx)
= 1 - \sum_{\bx \in \Omega^{(N)}} \pi^{(\infty)}(\bx). 
\label{eqn:trueerr}
\end{equation}

The true error $\Err^{(N)}$ is unknown, as it
requires knowledge of $\pi^{(\infty)}(\bx)$ for all $\bx \in
\Omega^{(\infty)}$.  In this section, we show that $\Err^{(N)}$
asymptotically converges to $\tilde{\pi}_N^{(\infty)}$ as the maximum 
net copy number limit $N$ increases. 
If $N$ is sufficiently large, the true error
$\Err^{(N)}$ is bounded by the true boundary probability
$\tilde{\pi}_N^{(\infty)}$ times a constant.  
First, we have:
\begin{lemma} (Finite Biological System.)
For any  biological system in which the total amount of mass is finite, the aggregated
synthesis rate $\alpha^{(\infty)}_{i}$ becomes smaller than the 
aggregated degradation rate $\beta^{(\infty)}_{i+1}$ 
when the total molecular copy number $N$ is sufficiently large:
\begin{equation}
\lim_{N \rightarrow \infty} \sup\limits_{i > N} \frac{\alpha^{(\infty)}_{i}}{\beta^{(\infty)}_{i+1}} < 1.
\label{eqn:fbs}
\end{equation}
\label{lm:fbs}
\end{lemma}

Proof can be found in the Appendix.

Note that in most biological reaction networks, the stronger condition 
$\lim_{N \rightarrow \infty} \sup\limits_{i > N} \frac{\alpha^{(\infty)}_{i}}{\beta^{(\infty)}_{i+1}} = 0$
should hold, as synthesis reactions usually have constant rates,
while degradation reactions have increasing rates when 
the copy number of the molecule increases.  
When the net copy number $i$ is sufficiently large, 
the ratio approaches zero.

According to the Eqn.~(\ref{eqn:trueerr}) and
$\tilde{\pi}^{(\infty)}_{i+1} < \tilde{\pi}^{(\infty)}_{i}$ as
discussed above, 
when the
total net molecular copy number $N$ increases to infinity,  
the true error $\Err^{(N)}$ converges to zero. 
For a
finite system, the series of the boundary probability
$\{{\tilde{\pi}^{(\infty)}_{N}}\}$ (Eqn.~(\ref{eqn:pininf})) also
converges to $0$, since the sequence of its partial sums converges to
$1$. That is, the $N$-th member $\ {\tilde{\pi}^{(\infty)}_{N}}$ of
this series converges to $0$ and the residual sum of this series $
\sum\limits_{i = N + 1}^\infty {\tilde{\pi}_i^{(\infty )}} \equiv
\Err^{(N)}$ converges to $0$.  We now study the convergence behavior
of the ratio of $\Err^{(N)}$ and $\tilde{\pi}^{(\infty)}_N$.
\begin{theorem} 
(Asymptotic Convergence of Error.)
For a truncated state space with a maximum net molecular copy number $N$ 
in the network, 
the true error  $\Err^{(N)}$ 
follows the inequality below when $N$ increases to infinity:
\begin{equation}
\Err^{(N)} \leq 
\frac{\frac{\alpha^{(\infty)}_{M}}{\beta^{(\infty)}_{M+1}}}{1-\frac{\alpha^{(\infty)}_{M}}{\beta^{(\infty)}_{M+1}}} \tilde{\pi}^{(\infty)}_N, 
\label{eqn:ace}
\end{equation}
where $M$ is an integer selected from $N, \cdots, \infty$ to satisfy 
$
\frac{\alpha^{(\infty)}_{M}}
     {\beta^{(\infty)}_{M+1}}
=\mathop {\sup }\limits_{k \ge N}
\left\{
   \frac{\alpha^{(\infty)}_{k}}
        {\beta^{(\infty)}_{k+1}}
\right\}$.
\label{thm:ace}
\end{theorem}

Proof can be found in the Appendix.

According to Theorem~\ref{thm:ace}, the true error $\Err^{(N)}$ is 
asymptotically bounded by the boundary probability $\tilde{\pi}^{(\infty)}_N$ 
multiplied by a simple function 
of the aggregated synthesis rates $\alpha^{(\infty)}_M$ and degradation 
rates $\beta^{(\infty)}_{M+1}$.  We can therefore use Inequality
(\ref{eqn:ace}) 
to construct an upper-bound for $\Err^{(N)}$.  We examine three cases: 
(1) If $\alpha^{(\infty)}_{M}/\beta^{(\infty)}_{M+1} < 0.5$, 
the true error is always smaller than the boundary
probability: $\Err^{(N)} < \pi^{(\infty)}_N$, 
when the maximum net molecular copy number $N$ is sufficiently large. 
(2) If $\alpha^{(\infty)}_{M}/\beta^{(\infty)}_{M+1} = 0.5$, the true error
converges asymptotically to $\pi^{(\infty)}_N$.  
(3) If $ 0.5 < \frac{\alpha^{(\infty)}_{N}}{\beta^{(\infty)}_{N+1}} < 1.0 $, 
the error is bounded by $\pi^{(\infty)}_N$ multiplied by a constant 
$
C \equiv 
\frac{
       \alpha^{(\infty)}_{M}/\beta^{(\infty)}_{M+1}
     }
{
       1-{\alpha^{(\infty)}_{M}}/{\beta^{(\infty)}_{M+1}}
}
$ according to  
Inequality~(\ref{eqn:ace}).

In realistic biological reaction networks, case (1) is most applicable.
As rates of synthesis reactions usually are constant, whereas
rates of degradation reactions depend on the copy number of net
molecules in the network, the ratio between aggregated synthesis rate
and degradation rate decreases monotonically with increasing net
molecular copy numbers $N$.  We therefore conclude that the boundary probability
$\pi_N^{(\infty)}$ indeed provides an upper bound to the state space
truncation error.  In addition, in case (1) $M = N$, and
$\alpha^{(\infty)}_{N}/\beta^{(\infty)}_{N+1} =
\alpha^{(\infty)}_{M}/\beta^{(\infty)}_{M+1}$.
Therefore Inequality~(\ref{eqn:ace}) can be further rewritten as: 
\begin{equation}
\Err^{(N)} \leq 
\frac{\frac{\alpha^{(\infty)}_{N}}{\beta^{(\infty)}_{N+1}}}{1-\frac{\alpha^{(\infty)}_{N}}{\beta^{(\infty)}_{N+1}}} \tilde{\pi}^{(\infty)}_N. 
\label{eqn:ace1}
\end{equation}

\subsubsection*{Computed Probability of Boundary States on Truncated State Space Bounds the True Boundary Probability}

It is not practical to compute the true boundary probability
  $\pi_N^{(\infty)}$ on the original infinite state space.  In this
  section, we show that the probability of boundary states
  $\pi^{(N)}_N$ is larger than $\pi^{(\infty)}_N$.  Therefore, we can
  use $\pi^{(N)}_N$ on truncated state space as an upper bound for
  $\Err^{(N)}$. That is, the steady state probability $\pi^{(N)}_{N}$
  computed using the truncated state space over the boundary states can be used to bound
$\Err^{(N)}$.

We first show that the truncated state space and its rate matrix can also be 
aggregated according to the net copy number of molecules in MEG 
following Lemma~\ref{lm:rmat}, which is similar to Lemma~\ref{lm:rma}: 

\begin{lemma}
\label{lm:rmat}
A Molecular Equivalent Group with a maximum of $N$ total copy number of elementary molecular species 
gives a truncated state space $\Omega^{(N)}$ and a truncated rate matrix $\bA^{(N)}$. 
For any homogeneous continuous-time Markov process with such a rate matrix $\bA^{(N)}$, 
an aggregated continuous-time Markov 
process with a rate matrix $\bB^{(N)}$ can be constructed
on the partition $\tilde{\Omega}^{(N)} = \{ \mathcal{G}_0,
\mathcal{G}_1, \cdots, \mathcal{G}_N \}$ with respect to 
the total net copy number of molecules in the network, such that it gives the same steady
state probability distribution for each partitioned group
$\{\mathcal{G}_s\}$ as that given by the original matrix $\bA^{(N)}$,
\textit{i.e.}, $\pi(\mathcal{G}_s) = \sum_{\bx \in \mathcal{G}_s}
\bpi(\bx)$ for all $s = 0, 1, \cdots, N$,
where $\bpi(\bx)$ is the steady state probability distribution associated with
$\bA^{(N)}$.

Specifically, the rate matrix $\bB^{(N)}$ can be constructed as: 
\begin{equation}
\bB^{(N)} = \left( {\begin{array}{c}
   \boldsymbol\alpha^{(N)}, \boldsymbol\gamma^{(N)}, \boldsymbol\beta^{(N)}
\end{array}} \right),
\label{eqn:bdmatf}
\end{equation}
with the lower off-diagonal vector $$\boldsymbol\alpha^{(N)} = (\alpha^{(N)}_i), \, i = 0,\cdots,N-1.$$
the upper off-diagonal vector $$\boldsymbol\beta^{(N)} = (\beta^{(N)}_{i}), \, i = 1,\cdots,N.$$
and the diagonal vector $$\boldsymbol\gamma^{(N)} = (\gamma^{(N)}_i) = (-\alpha^{(N)}_i - \beta^{(N)}_{i}), \, i = 0,\cdots,N.$$
It is equivalent to substituting 
the block sub-matrices $\bA_{i+1,\, i}$ and $\bA_{i,\, i+1}$ in the 
original rate matrix $\tilde{\bA}$ with the 
corresponding aggregated synthesis rate $\alpha^{(N)}_i$ and 
degradation rate $\beta^{(N)}_{i+1}$, respectively. 
The aggregated rates on the truncated state space are: 
\begin{equation}
\alpha^{(N)}_i = \mathbbm{1}^T \bA_{i+1,i} \frac{
\tilde{\bpi}(\mathcal{G}_{i})}{\mathbbm{1}^T \tilde{\bpi}(\mathcal{G}_{i})} \quad \text{and} \quad
\beta^{(N)}_{i+1} = \mathbbm{1}^T \bA_{i,i+1} \frac{\tilde{\bpi}(\mathcal{G}_{i+1})}{\mathbbm{1}^T \tilde{\bpi}(\mathcal{G}_{i+1})}.
\label{eqn:abdeff}
\end{equation}

\end{lemma}

\begin{proof}
Same as Lemma~\ref{lm:rma}. 
\end{proof}

Similar to the case of infinite state space, we can write out in  analytic form 
the total steady state probability 
$\tilde{\pi}_i^{(N)}$
over each aggregated group
$\mathcal{G}_i$ as: 
\begin{equation}
\tilde{\pi}_i^{(N)} = \frac{\prod\limits_{k = 0}^{i-1}
    \frac{\alpha^{(N)}_{k}}{\beta^{(N)}_{k+1}}}
{1 + {\sum\limits_{j = 1}^{N} {\prod\limits_{k = 0}^{j-1}
     {\frac{\alpha^{(N)}_{k}}{\beta^{(N)}_{k+1}}}} 
    }
}, \quad i = 0, 1, 2, \cdots, N. 
\label{eqn:pini}
\end{equation}
Specifically, the total steady state probability $\tilde{\pi}_N^{(N)}$
over the group of aggregated boundary states $\mathcal{G}_N$ is: 
\begin{equation}
\tilde{\pi}_N^{(N)} = \frac{\prod\limits_{k = 0}^{N-1}
    \frac{\alpha^{(N)}_{k}}{\beta^{(N)}_{k+1}}}
{1 + {\sum\limits_{j = 1}^{N} {\prod\limits_{k = 0}^{j-1}
     {\frac{\alpha^{(N)}_{k}}{\beta^{(N)}_{k+1}}}} 
    }
}. 
\label{eqn:pinn}
\end{equation}

We now study how state space truncation
affects the steady state probabilities over the aggregated groups. 
\begin{theorem}
\label{thm:ipt}
(Boundary Probability Increases after State Space Truncation)
The total steady state probability $\tilde{\pi}_i^{(N)}$ of an aggregated state 
group $\mathcal{G}_i$, for all $i = 0, 1, \cdots, N$, on the truncated state space 
$\tilde{\Omega}^{(N)}$ with a maximum net molecular copy number $N$, 
is greater than or equal to the non-truncated probability 
$\tilde{\pi}_i^{(\infty)}$ over the same group $\mathcal{G}_i$
obtained using the original 
state space $\tilde{\Omega}^{(\infty)}$ of infinite size, 
i.e., $\tilde{\pi}_i^{(\infty)} \leq \tilde{\pi}_i^{(N)}$. 
\end{theorem}

Proof can be found in the Appendix.

In summary, the boundary probability increases when the state space is
truncated $\tilde{\pi}_N^{(N)} \geq \tilde{\pi}_N^{(\infty)}.$ From
Theorem~\ref{thm:ace}, we always have $\Err^{(N)} \leq C
\tilde{\bpi}_N^{(\infty)}$.  Therefore, we can bound $\Err^{(N)}$ by
the boundary probability $\pi^{(N)}_N$ computed using the truncated
state space when $\alpha^{(N)}_i \ne 0$ and $\beta^{(N)}_{i+1} \ne 0$.

\subsubsection*{From One to Multiple MEGs}

In complex reaction networks, multiple MEGs occur.  Since different
MEGs are pairwise disjoint, we can aggregate the same state space and
re-construct the permuted the rate matrix according to different MEG one at a time.
Lemmas~\ref{lm:rma}, \ref{lm:fbs}, and \ref{lm:rmat}, and
Theorems~\ref{thm:ace} and \ref{thm:ipt} are all valid for each
individual MEG.  That is, the true error of truncating one MEG is
bounded by the boundary probability computed using the state space
truncated in that particular MEG, while all other MEGs have infinite
net molecular copy numbers.  However, it is not possible to compute
the solution of dCME with infinite molecules in any MEG. Below we
study how error bounds can be constructed when states in all MEGs are
truncated simultaneously.

\subsubsection*{From Truncating One to Truncating All MEGs}

We use $\mathcal{I} = (\infty, \cdots, \infty)$ to denote the vector of  
infinite  net copy numbers for all MEGs in the network. 
$\mathcal{I}$ corresponds to the original infinite 
state space $\Omega^{(\mathcal{I})}$ without any truncation. 
We use $\bA^{(\mathcal{I})}$ and $\bpi^{(\mathcal{I})}$ to 
denote the transition rate matrix and the steady state probability 
distribution over $\Omega^{(\mathcal{I})}$, respectively. Furthermore, we have 
$\bA^{(\mathcal{I})} \bpi^{(\mathcal{I})} = 0$.

We use $\mathcal{I}_j = (\infty, \cdots, N_j, \cdots, \infty)$ to
denote the vector of maximum copy numbers with only the $j$-th MEG
limited to a finite copy number $N_j$ and all other MEGs with
infinite copy numbers.  The corresponding state space is
denoted  $\Omega^{(\mathcal{I}_j)}$, the transition rate matrix
$\bA^{(\mathcal{I}_j)}$, and the steady state probability distribution
$\bpi^{(\mathcal{I}_j)}$. At the steady state, we also have
$\bA^{(\mathcal{I}_j)} \bpi^{(\mathcal{I}_j)} = 0$.

We now add one more truncation to the $i$-th MEG in addition to the
$j$-th MEG.  We denote the vector of maximum copies
as $\mathcal{I}_{i,j}
= (\infty, \cdots, N_i, \cdots, N_j, \cdots, \infty)$,
with $N_i$ and $N_j$ the maximum copy numbers of the $i$-th and $j$-th MEG, 
respectively. All other MEGs can have infinite molecular copy numbers.  We denote the 
corresponding state space as $\Omega^{(\mathcal{I}_{i,j})}$, the transition 
rate matrix  $\bA^{(\mathcal{I}_{i,j})}$, the steady state probability 
distribution  $\bpi^{(\mathcal{I}_{i,j})}$.  At the steady state, we have 
$\bA^{(\mathcal{I}_{i,j})} \bpi^{(\mathcal{I}_{i,j})} = 0$. 

When all $w$ number of MEGs in the network are truncated using a vector of maximum copies  
$\mathcal{B} = (N_1, \cdots, N_i, \cdots, N_j, \cdots, N_w)$, we have a finite 
state space $\Omega^{(\mathcal{B})}$. 
Obviously, we have 
$\Omega^{(\mathcal{B})} \subseteq 
\Omega^{(\mathcal{I}_{i,j})} \subseteq 
\Omega^{(\mathcal{I}_j)} \subseteq 
\Omega^{(\mathcal{I})}$.

We have already shown that for each truncated MEG on the infinite
state space, the truncation error is bounded by the corresponding
boundary probability.  We now show that this error bound also holds
for the fully truncated state spaces $\Omega^{(\mathcal{B})}$. We
show first adding only one additional truncation at the $i$-th MEG
to the singularly truncated state space $\Omega^{(\mathcal{I}_j)}$,
and demonstrate that the probability of each state in the doubly
truncated state space $\Omega^{(\mathcal{I}_{i,j})}$ is no smaller
than the probability in singularly truncated state space
$\Omega^{(\mathcal{I}_j)}$, \textit{i.e.},
$\pi^{(\mathcal{I}_{i,j})}(\bx) \geq \pi^{(\mathcal{I}_{j})}(\bx)$ for
all $\bx \in \Omega^{(\mathcal{I}_{i,j})}$.

\begin{theorem}
\label{thm:Iij}
At steady state, $\bpi^{(\mathcal{I}_{i,j})} \geq \bpi^{(\mathcal{I}_{j})}$ and 
$\bpi^{(\mathcal{I}_{i,j})}$ approaches $\bpi^{(\mathcal{I}_{j})}$ component-wise 
for any state in $\Omega^{(\mathcal{I}_{i,j})}$ when the maximum net copy number limit for 
the $i$-th MEG $N_i$ goes to $\infty$. 
\end{theorem}

Proof can be found in the Appendix.

Theorem~\ref{thm:Iij} shows that introducing an additional truncation 
at the $i$-th MEG does not decrease the boundary probability of the $j$-th MEG. 
Therefore, the boundary probability from doubly truncated state space 
$\Omega^{(\mathcal{I}_{i,j})}$ can also be used to bound the 
true error after state truncations at both $i$-th and $j$-th MEG. 
Furthermore, we can show by induction that boundary probabilities computed 
from the fully truncated state space $\Omega^{(\mathcal{B})}$ can also be 
used to bound the truncation errors of each MEG, respectively.

\subsubsection*{Upper and Lower Bounds for Steady State Boundary Probability}

In this section, we introduce an efficient and easy-to-compute method
to obtain an upper- and lower-bound of the boundary probabilities
$\tilde{\pi}^{(N)}_N$ {\it a priori}\/ without the need to solving the
dCME.  The method can be used to rapidly determine if the maximum 
copy number limits to MEGs are adequate to obtain the direct solution
to dCME with a truncation error smaller than the predefined tolerance.
The optimal maximum copy number for each MEG can therefore be
estimated \textit{a priori}.

As a consequence of Theorem~(\ref{thm:Iij}) discussed above, the
boundary probability computed on the truncated state space
$\Omega^{(B)}$ can be used as an error bound. We now use the truncated
rate matrix to derive the upper- and lower-bounds.

Denote the maximum and minimum aggregated synthesis rates from the block 
sub-matrix $\bA_{i+1,\,i}$  as 
\begin{equation}
\overline{\alpha}^{(N)}_i = \max\{\mathbbm{1}^T \bA_{i+1,i}\} \quad \mbox{ and } \quad
\underline{\alpha}^{(N)}_i = \min\{\mathbbm{1}^T \bA_{i+1,i}\},
\label{eqn:alu}
\end{equation}
respectively, and the maximum and minimum aggregated degradation rates
from the block sub-matrix $\bA_{i,\,i+1}$ as
\begin{equation}
\overline{\beta}^{(N)}_{i+1} = \max\{\mathbbm{1}^T \bA_{i,i+1}\} \quad \mbox{ and } \quad
\underline{\beta}^{(N)}_{i+1} = \min\{\mathbbm{1}^T \bA_{i,i+1}\},
\label{eqn:blu}
\end{equation}
respectively.  Note that $\overline{\alpha}^{(N)}_i$,
$\underline{\alpha}^{(N)}_i$, $\overline{\beta}^{(N)}_{i+1}$, and
$\underline{\beta}^{(N)}_{i+1}$ can be easily calculated from the 
reaction rates in the network without need for generating and 
partitioning the dCME transition rate matrix $\tilde{\bA}$. 
As $\alpha^{(N)}_i$ and $\beta^{(N)}_{i+1}$ given in Eqn.~(\ref{eqn:abdef})
are weighted sums of vector $\mathbbm{1}^T \bA_{i+1,i}$ and
$\mathbbm{1}^T \bA_{i,i+1}$ with regard to the steady state
probability distribution $\tilde{\bpi}^{(N)}(\mathcal{G}_i)$, 
respectively, we have 
$$
\underline{\alpha}^{(N)}_i \leq \alpha^{(N)}_i \leq \overline{\alpha}^{(N)}_i \quad \mbox{ and } \quad
\underline{\beta}^{(N)}_{i+1} \leq \beta^{(N)}_{i+1} \leq \overline{\beta}^{(N)}_{i+1}. 
$$

We use results from the theory of stochastic ordering 
for comparing Markov processes to bound  $\pi^{(N)}_N$. Stochastic ordering 
``$\leq_{st} $''
between two infinitesimal generator 
matrices $\bP_{n \times n}$ and $\bQ_{n \times n}$ of 
Markov processes is defined as~\cite{Truffet1997,Irle2003} 
$$
\bP \leq_{st} \bQ \quad \text{if and only if } \sum_{k=j}^n P_{i,k} \leq \sum_{k=j}^n Q_{i,k} \text{ for all } i, j. 
$$ To derive an upper bound for $\tilde{\pi}^{(N)}_N$ in
Eqn.~(\ref{eqn:pinn}), we construct a new matrix $\overline{\bB}$ by
replacing $\alpha^{(N)}_{i}$ with the corresponding
$\overline{\alpha}^{(N)}_i$ and $\beta^{(N)}_{i+1}$ with the
corresponding $\underline{\beta}^{(N)}_{i+1}$ in the matrix $\bB$.
Similarly, to derive an lower bound for $\tilde{\pi}^{(N)}_N$, we construct the
matrix $\underline{\bB}$ by replacing $\alpha^{(N)}_{i}$ with the
corresponding $\underline{\alpha}^{(N)}_i$ and replace
$\beta^{(N)}_{i+1}$ with $\overline{\beta}^{(N)}_{i+1}$ in $\bB$.  
We then have the following stochastic ordering:
$$
\underline{\bB} \leq_{st} \bB \leq_{st} \overline{\bB}. 
$$ 
All three matrices $\underline{\bB}$, $\bB$, and $\overline{\bB}$
are ``$\leq_{st}-\rm{monotone}$'' according to the definitions in
Truffet~\cite{Truffet1997}.  The steady state probability
distributions of matrices $\underline{\bB}$, $\bB$, and
$\overline{\bB}$ maintain the same stochastic ordering (Theorem 4.1 of
Truffet~\cite{Truffet1997}):
$$
\bpi_{\underline{\bB}} \leq_{st} \bpi_{\bB} \leq_{st} \bpi_{\overline{\bB}}. 
$$ 
Therefore, we have the inequality:
$$
\underline{\tilde{\pi}}^{(N)}_N
\leq 
\tilde{\pi}^{(N)}_N
\leq 
\overline{\tilde{\pi}}^{(N)}_N. 
$$ 
Here 
the lower bound
$\underline{\tilde{\pi}}^{(N)}_N$ is the boundary probability from
$\bpi_{\underline{\bB}}$, 
$\tilde{\pi}^{(N)}_N$ is the boundary
probability from $\bpi_{\bB}$,
and 
the upper bound $\overline{\tilde{\pi}}^{(N)}_N$ is the boundary
probability computed from  $\bpi_{\overline{\bB}}$. 
From Eqn.~(\ref{eqn:pinn}), the upper
bound $\overline{\tilde{\pi}}^{(N)}_N$ can be calculated \textit{a priori}\/ 
from reaction rates:
\begin{equation}
\overline{\tilde{\pi}}^{(N)}_N = 
\frac{\prod\limits_{k = 0}^{N-1}
    \frac{\overline{\alpha}^{(N)}_{k}}{\underline{\beta}^{(N)}_{k+1}}} 
{1 + {\sum\limits_{j = 1}^{N} {\prod\limits_{k = 0}^{j-1}
     {\frac{\overline{\alpha}^{(N)}_{k}}{\underline{\beta}^{(N)}_{k+1}}}} 
    }
}, 
\label{eqn:upb}
\end{equation}
and the lower bound $\underline{\tilde{\pi}}^{(N)}_N$ can also be calculated as:
\begin{equation}
\underline{\tilde{\pi}}^{(N)}_N = 
\frac{\prod\limits_{k = 0}^{N-1}
    \frac{\underline{\alpha}^{(N)}_{k}}{\overline{\beta}^{(N)}_{k+1}}}
{1 + {\sum\limits_{j = 1}^{N} {\prod\limits_{k = 0}^{j-1}
     {\frac{\underline{\alpha}^{(N)}_{k}}{\overline{\beta}^{(N)}_{k+1}}}} 
    }
}. 
\label{eqn:lowerb}
\end{equation} 
These are general formula for upper and lower bounds of the boundary
probabilities of any MEG in a reaction network.
Note that while $\overline{\tilde{\pi}}^{(N)}_N$ is easy to compute,
it may not be a tight error bound when the MEG involves many molecular
species with overall complex interactions.  This will be shown in the
example of the phage lambda epigenetic switch model
(Fig.~\ref{fig:ph1}A and B).

For a reaction network
with multiple MEGs, we have 
$$
\sum_{i=1}^w \underline{\tilde{\pi}}^{(N_i)}_{N_i}
\leq 
\sum_{i=1}^w \tilde{\pi}^{(N_i)}_{N_i}
\leq 
\sum_{i=1}^w \overline{\tilde{\pi}}^{(N_i)}_{N_i},
$$ where $N_i$ is the maximum copy number for the $i$-th MEG.
The upper bounds for the total error 
$\Err^{(\Omega^{(\mathcal{B})})}$ 
can therefore 
 be obtained straightforwardly by taking
summation of upper bounds for each individual MEG:
\begin{equation}
\Err^{(\Omega^{(\mathcal{B})})} 
\leq
\sum_{i=1}^w \overline{\tilde{\pi}}^{(N_i)}_{N_i},
\end{equation}
This upper bound of
$\sum_{i=1}^w \overline{\tilde{\pi}}^{(N_i)}_{N_i}$ can therefore be used as an
\textit{a priori} estimated bound for the total truncation error 
$\Err^{(\Omega^{(\mathcal{B})})}$
for the state
space $
\Omega^{(\mathcal{B})} 
$
 using truncation of $\mathcal{B} = (N_1, \cdots, N_i, \cdots, N_j,
\cdots, N_w)$.

\newpage

\section*{Biological Examples}
\label{sec:ex}
Below we give examples on characterizing the truncation errors in the
steady state probability landscapes for four biological reaction
networks. We study the models of the birth and death process, the 
single gene expression, the model of genetic toggle switch, and the 
phage lambda epigenetic switch model.  We first show how each network
can be partitioned into MEGs, and how truncation errors for each MEG
can be estimated \textit{a priori}. By enumerating the state space and
directly computing the steady state probability landscapes of the
dCMEs using the fb-dCME method, we examine the true truncation errors,
the computed boundary probabilities, and the {\it a priori}\/
estimated truncation error. We demonstrate that indeed the truncation
error is bounded from above by the computed boundary probability, and 
by the \textit{a priori} error estimate according to theoretical
analyses described earlier, once the copy number limit is sufficiently
large for the MEG(s).

\subsection*{Birth-Death Process}
The birth-death process is a ubiquitous biochemical phenomenon.  In
its simplest form, it involves synthesis and degradation of only one
molecular species.  We study this simple birth-death process, whose
reaction scheme and rate constants are specified as follows:
\begin{equation}
\label{eqn:bdrxns}
\begin{split}
&R_1: \quad \emptyset \stackrel{k_s}{\rightarrow} X, \quad k_s = 1 /s, \\
&R_2: \quad X \stackrel{k_d}{\rightarrow} \emptyset, \quad k_d = 0.025 /s. \\
\end{split}
\end{equation}
The steady state probability landscape of the birth-death process is
well known~\cite{Taylor1998,Vellela2007}. This process has also been 
studied extensively as a problem of estimating rare event
probability~\cite{Daigle2011,Roh2011,Cao2013JCP}.

\paragraph{Molecular equivalent group (MEG). }
This single birth and death process is an open network because of the
presence of the synthesis reaction. There is only one molecular
equivalent group (MEG).  We truncate the state space at different
values of the maximum copy number of the MEG, ranging from $0$ to
$200$, and compute the boundary probabilities at each different
truncation.

\paragraph{Asymptotic convergence of errors (Theorem~\ref{thm:ace}). }
To numerically demonstrate Theorem~\ref{thm:ace}, we 
compute the true truncation error of the steady state solution to 
the dCME.  We use a large copy number of $MEG=200$, which 
gives an infinitesimally small boundary probability of $1.391 \times 10^{-72}$. 
Steady state solution obtained using this MEG number 
coincides with analytical solution, and is
therefore considered to be exact.  
With this exact steady state probability landscape, the true truncation error
$\Err^{(N)}$ at smaller MEG sizes can be computed using
Eqn.~(\ref{eqn:trueerr}) (Fig.~\ref{fig:bd1}A, blue dashed line and crosses).  
The corresponding boundary probabilities $\pi^{(\infty)}_N$ 
are computed from this exact steady state probability landscape 
(Fig.~\ref{fig:bd1}A, green dashed line and circles). 

Consistent with the statement in the Theorem~\ref{thm:ace}, 
we find here that the true error $\Err^{(N)}$ 
(Fig.~\ref{fig:bd1}A, blue dashed line and crosses) is 
bounded by the computed boundary probability $\pi^{(\infty)}_N$ 
(Fig.~\ref{fig:bd1}A, green dashed line and circles) 
when the size of the MEG is sufficiently large. 
The inset of Fig.~\ref{fig:bd1}A 
shows the ratio of the true errors to the computed errors at different 
sizes of the MEG, and the grey straight line marks the ratio one. 
The computed errors are larger than the true errors when the black line 
is below the grey straight line (Fig.~\ref{fig:bd1}A inset).
In this example, 
the computed boundary probability is greater than the 
true error when $N > 79$, as would
be expected from Theorem~\ref{thm:ace}.

\paragraph{\textit{A priori} estimated error bound. }
To examine the \textit{a priori} estimated upper bound for truncation
error, we follow Eqn.~(\ref{eqn:alu}) and (\ref{eqn:blu}) to assign
values of $\overline{\alpha}_{i} = k_s$ and $\underline{\beta}_{i+1} =
k_d (i+1)$ for this network.  We compute the \textit{a priori} upper error
bound for different truncations using Eqn.~(\ref{eqn:upb})
(Fig.~\ref{fig:bd1}A, red solid line).  For this simple network,
$\overline{\alpha}_i = \alpha_i = \underline{\alpha}_i$ and
$\overline{\beta}_i = \beta_i = \underline{\beta}_i$, therefore the
\textit{a priori} estimated error is exactly the same as the analytic
solution for the steady state distribution for this simple birth-death
network, and it coincides with the computed error (Fig.~\ref{fig:bd1}A
red and green lines).  The true error, computed error, and the
\textit{a priori} error bound all decrease monotonically with
increasing MEG size $N$ (Fig.~\ref{fig:bd1}A).

\paragraph{Increased probability after state space truncation (Theorem~\ref{thm:ipt}). }
According to Theorem~\ref{thm:ipt}, the probability of a state
increases upon state space truncation.  We compare the steady state
probability landscapes of $X$ computed using truncations at different
sizes ranging from 40 to 50 with the exact steady state landscape
(Fig.~\ref{fig:bd1}B, red line). Our results indeed show clearly that
all probabilities increase as more states are truncated
(Fig.~\ref{fig:bd1}B).  The probability landscape computed using $N =
50$ (Fig.~\ref{fig:bd1}B, yellow line) or larger is very close to the
exact landscape using $N =200$ (Fig.~\ref{fig:bd1}B, red line).
However, the probability landscapes computed using smaller $N$ deviate
significantly from the exact probability landscape. The smaller the
MEG size, the more significant the deviation is.  These results are
fully consistent with the statements of Theorem~\ref{thm:ipt}.

\subsection*{Single Gene Expression Model}
Transcription and translation are fundamental processes in gene
regulatory networks that often involve significant stochasticity.  The
abundance of mRNA and expressed proteins of a gene is usually 2--4
orders of magnitude apart in a cell. There are only a few or dozens of copies of mRNA molecules in
each cell for one gene, but the copy number of proteins can range from
hundreds to ten thousands~\cite{Taniguchi2010}. Here we study a model
of the fundamental process of single gene transcription and
translation using the following reaction scheme and rate constants:
\begin{equation}
\label{eqn:bdrxns2}
\begin{split}
&R_1: \quad Gene + \emptyset \stackrel{k_e}{\rightarrow} Gene + mRNA, \quad k_e = 1.0 /s, \\
&R_2: \quad mRNA + \emptyset \stackrel{k_t}{\rightarrow} mRNA + Protein, \quad k_t = 1.0 /s. \\
&R_3: \quad mRNA \stackrel{k_m}{\rightarrow} \emptyset, \quad k_m = 0.1 /s. \\
&R_4: \quad Protein \stackrel{k_d}{\rightarrow} \emptyset, \quad k_d = 0.01 /s. \\
\end{split}
\end{equation}

\paragraph{Molecular equivalent group (MEG). }
This single gene expression model is an open network. We can
participate this model into two molecular equivalent groups (MEG),
with MEG$_1$ consists of species $mRNA$, MEG$_2$ consists of
$Protein$.  Note that protein synthesis depends on the copy number
$mRNA$, despite the fact that $mRNA$ and $Protein$ are two independent
molecular species that cannot be transformed into each other.

\paragraph{Asymptotic convergence of errors (Theorem~\ref{thm:ace}). }
To numerically demonstrate Theorem~\ref{thm:ace}, we compute the true
error of the steady state solution to the dCME using sufficiently
large sizes of $MEG_1 = 64$ and $MEG_2 = 2,580$, which gives negligible
truncation error, with infinitesimally small boundary probabilities
$3.58 \times 10^{-30}$ for $MEG_1$ and $1.15 \times 10^{-32}$ for
$MEG_2$.  Solution obtained using these MEGs is therefore considered
to be exact.  With this exact steady state probability landscape, the
true truncation error $\Err^{(N)}$ at smaller sizes of $MEG_1$ and
$MEG_2$ can be computed using Eqn.~(\ref{eqn:trueerr})
(Fig.~\ref{fig:sge1}A and B, blue dashed lines and crosses).  The
corresponding boundary probabilities $\pi^{(\infty)}_N$ or computed
error are obtained from the exact steady state probability landscape
for both $MEG_1$ (Fig.~\ref{fig:sge1}A, green dashed line and circles)
and $MEG_2$ (Fig.~\ref{fig:sge1}B, green dashed line and circles).

Consistent with the statement in Theorem~\ref{thm:ace}, 
our results show that the true error $\Err^{(N)}$ 
is bounded by the computed boundary probability $\pi^{(\infty)}_N$ 
in the $MEG_1$ when $N_1 \geq 20$ (Fig.~\ref{fig:sge1}A, blue dashed lines and
crosses, green dashed lines and circles, and the inset). 
In the $MEG_2$, although the true errors are larger than computed errors 
even when the MEG size is large (Fig.~\ref{fig:sge1}A inset), 
the true error can be bounded by the computed error
when $N_2 \geq 5000$ by a multiplication factor of 6
(Fig.~\ref{fig:sge1}B, blue dashed lines and crosses, green dashed
lines and circles, and the inset).  This is expected from
Theorem~\ref{thm:ace}.

\paragraph{\textit{A priori} estimated error bound. }
To examine \textit{a priori} estimated upper bounds for the truncation
errors in MEG$_1$ and MEG$_2$, we follow Eqn.~(\ref{eqn:alu}) and
(\ref{eqn:blu}) to assign values of $\overline{\alpha}_{i}=k_{e}$ and
$\underline{\beta}_{(i+1)}=k_{m} (i+1)$ for the $MEG_1$.  Because of
the dependency of protein synthesis on the mRNA copy numbers, we set
$\overline{\alpha}_{i}=64 \cdot k_{t}$ and $\underline{\beta}_{i+1}=
k_{d} (i+1)$ following Eqn.~(\ref{eqn:alu}) and (\ref{eqn:blu}) for
the $MEG_2$, where the factor $64$ is the maximum copy number of mRNA
in the $MEG_1$.  We compute the \textit{a priori} estimated upper
bounds of errors for different truncations of MEG$_1$ and MEG$_2$
using Eqn.~(\ref{eqn:upb}) (Fig.~\ref{fig:sge1}A and B, red solid
lines).  The true truncation errors and the \textit{a priori}
estimated error bounds of MEG$_1$ and MEG$_2$ all decrease
monotonically with increasing MEG sizes (Fig.~\ref{fig:sge1}A and B).
The computed errors also monotonically decrease in both MEGs.  For
$MEG_1$, the {\it a priori}\/ estimated error bounds coincide with the
computed errors (Fig.~\ref{fig:sge1}A red and green lines).  For the
$MEG_2$, the {\it a priori}\/ estimated error bounds are larger than
computed errors at all MEG sizes.

\paragraph{Increased probability after state space truncation (Theorem~\ref{thm:ipt}). }
According to Theorem~\ref{thm:ipt}, the probability landscape
projected on the MEGs increase after state space truncation. We 
compute the steady state probability landscapes of $Protein$ obtained using
truncations at different sizes of the MEG, ranging from $0$ to $2,580$
for $MEG_2$ while $MEG_1$ is fixed at 64 (Fig.~\ref{fig:sge1}C).  The
results are compared with the exact steady state landscape computed
using $MEG_2 = 2,600$ (Fig.~\ref{fig:sge1}C, red line).

Our results show clearly that all probabilities in the landscapes
increase when more states are truncated at smaller MEG size
(Fig.~\ref{fig:sge1}C).  The probability landscapes computed using
larger size of the MEG (\textit{e.g.}, $MEG_2 = 1400$,
Fig.~\ref{fig:sge1}C, yellow line) are approaching the exact landscape
(Fig.~\ref{fig:sge1}C, red line).  The probability landscapes obtained
using smaller MEG sizes deviate significantly from the exact
probability landscape. The smaller the MEG size, the more pronounced
the deviation is.  These numerical results are fully consistent with
Theorem~\ref{thm:ipt}.

\paragraph{Truncating additional MEGs does not decrease probabilities (Theorem~\ref{thm:Iij}). }
We further examine Theorem~\ref{thm:Iij}, \textit{i.e.}, the
probability landscape projected on one MEG increase with state space
truncation at another MEG.  We compare the projected steady
state probability landscapes on $mRNA$ obtained using truncations of
different sizes of $MEG_2$ ranging from $0$ to $2580$ while
the $MEG_1$ is fixed at 64 (Fig.~\ref{fig:sge1}D). We compare the results
with the exact steady state landscape (Fig.~\ref{fig:sge1}D, red line).

Our results show that all probabilities on the landscapes 
of $mRNA$ are not affected by the truncations at the 
$MEG_2$ (Fig.~\ref{fig:sge1}D).  
The probability landscapes computed using different sizes of $MEG_2$ 
are the same (Fig.~\ref{fig:sge1}D).  
These numerical results are completely 
consistent with Theorem~\ref{thm:Iij}, because the probabilities of $mRNA$ 
are not decreased by the truncation at the MEG of $Protein$.

\subsection*{Genetic Toggle Switch}
The bistable genetic toggle switch consists of two genes repressing
each other through binding of their protein dimeric products on the
promoter sites of the other genes.  This genetic network has been
studied
extensively~\cite{Gardner2000,Kepler2001,Kim2007,Schultz_JCP07}.  We
follow references~\cite{Schultz_JCP07,CaoBMCSB08} and study a detailed
model of the genetic toggle switch with a more realistic
control mechanism of gene regulations.  Different from simpler toggle
switch models~\cite{Munsky2008,Deuflhard2008,Sjoberg2009,Kazeev2014},
in which gene binding and unbinding reactions are approximated by Hill
functions, here details of the gene binding and unbinding reactions
are modeled explicitly.  The molecular species, reactions, and their
rate constants are listed below:
\begin{equation}
\label{eqn:tgrxns}
\begin{split}
&R_1: GeneA \stackrel{k_1}{\rightarrow} GeneA + A, \quad k_{sA} = 40 \, s^{-1} \\
&R_2: GeneB \stackrel{k_2}{\rightarrow} GeneB + B, \quad k_{sB} = 20 \, s^{-1} \\
&R_3: A \stackrel{k_3}{\rightarrow} \emptyset, \quad k_{dA} = 1 \, s^{-1} \\
&R_4: B \stackrel{k_4}{\rightarrow} \emptyset, \quad k_{dB} = 1 \, s^{-1} \\
&R_5: 2A + GeneB \stackrel{k_5}{\rightarrow} bGeneB, \quad k_{bA} = 1 \times 10^{-5} \, nM^{-2} \cdot s^{-1} \\
&R_6: 2B + GeneA \stackrel{k_6}{\rightarrow} bGeneA, \quad k_{bB} = 3.5 \times 10^{-5} \, nM^{-2} \cdot s^{-1} \\
&R_7: bGeneB \stackrel{k_7}{\rightarrow} 2A + GeneB, \quad k_{uA} = 1 \, s^{-1} \\
&R_8: bGeneA \stackrel{k_8}{\rightarrow} 2B + GeneA, \quad k_{uB} = 1 \, s^{-1} \\
\end{split}
\end{equation}
Specifically, two genes $GeneA$ and $GeneB$ express protein products
$A$ and $B$, respectively.  Two protein monomers $A$ or $B$ can bind on
the promoter site of $GeneB$ or $GeneA$ to form protein-DNA complexes
$bGeneB$ or $bGeneA$, and turn off the expression of $GeneB$ or $GeneA$, 
respectively.

\paragraph{Molecular equivalent group (MEG). }
There are two MEGs in this network, MEG$_1$ consists of 
species $A$ and $bGeneB$, MEG$_2$ consists of $B$ and $bGeneA$. 

\paragraph{Asymptotic convergence of errors (Theorem~\ref{thm:ace}). }
To numerically demonstrate Theorem~\ref{thm:ace}, we compute the true
error of the steady state solution to the dCME using sufficiently
large sizes of $MEG_1 = 120$ and $MEG_2 = 80$, which gives negligible
truncation error, with infinitesimally small boundary probabilities
$5.275 \times 10^{-24}$ for $MEG_1$ and $2.561 \times 10^{-23}$ for
$MEG_2$.  Solution obtained using these MEGs is therefore considered
to be exact.  With this exact steady state probability landscape, the
true truncation error $\Err^{(N)}$ at smaller sizes of $MEG_1$ and
$MEG_2$ can both be computed using Eqn.~(\ref{eqn:trueerr})
(Fig.~\ref{fig:tg1}A and B, blue dashed lines and crosses).  The
corresponding boundary probabilities $\pi^{(\infty)}_N$ or computed
error are computed from the exact steady state probability landscape
for both $MEG_1$ (Fig.~\ref{fig:tg1}A, green dashed line and circles)
and $MEG_2$ (Fig.~\ref{fig:tg1}B, green dashed line and circles).

Consistent with the statement in Theorem~\ref{thm:ace}, 
our results show that the true error $\Err^{(N)}$ 
(Fig.~\ref{fig:tg1}A and B, blue dashed lines and crosses) is 
bounded by the computed boundary probability $\pi^{(\infty)}_N$ 
(Fig.~\ref{fig:tg1}A and B, green dashed lines and circles) 
when the size of the MEG is sufficiently large. 
The insets in Fig.~\ref{fig:tg1}A  and B
show the ratios of the true errors to the computed errors at different 
sizes of the MEG, and the grey straight line marks the ratio one. 
The computed errors are larger than the true errors when the black line 
is below the grey straight line (Fig.~\ref{fig:tg1}A and B, insets).
In this example, 
the computed boundary probability is greater than the 
true error when $MEG_1 > 82$ and $MEG_2 > 42$, 
as would be expected from Theorem~\ref{thm:ace}.

\paragraph{\textit{A priori} estimated error bound. }
To examine \textit{a priori} estimated upper bounds for 
the truncation errors in MEG$_1$ and MEG$_2$, 
we follow Eqn.~(\ref{eqn:alu}) and (\ref{eqn:blu}) to assign values of 
$\overline{\alpha}_{i}=k_{sA}$ and
$\underline{\beta}_{(i+1)}=[(i+1)-2]\cdot k_{dA}$ for the $MEG_1$,
where the subscript $(i+1)$ is the total copy number of species $A$ in
the system. The subtraction of $2$ is necessary because up to $2$
copies of $A$ can be protected from degradation by binding to $GeneB$.
This corresponds to the extreme case when $GeneA$ is constantly turned
on and $GeneB$ is constantly turned off. 
Similarly, we have
$\overline{\alpha}_{i}=k_{sB}$ and $\underline{\beta}_{i+1}=[(i+1)-2]
\cdot k_{dB}$ following Eqn.~(\ref{eqn:alu}) and (\ref{eqn:blu}) for the
$MEG_2$. This corresponds to the other extreme case when the
$GeneB$ is constantly turned on, and $GeneA$ is constantly turned off. 
We compute the \textit{a priori} estimated upper bounds of errors 
for different truncations of MEG$_1$ and MEG$_2$ using Eqn.~(\ref{eqn:upb})
(Fig.~\ref{fig:tg1}A and B, red solid lines).  
The true truncation errors and the \textit{a priori} estimated error bounds of MEG$_1$ and MEG$_2$ all 
decrease monotonically with increasing MEG sizes (Fig.~\ref{fig:tg1}A and B).  
The computed errors also monotonically decrease when the MEG sizes are larger 
than $40$ for MEG$_1$ and $20$ for MEG$_2$.  For both MEGs,
the {\it a priori}\/ estimated error bounds are larger
than computed errors at all MEG sizes.  They are also larger than the true
errors when the MEG sizes are sufficiently large.

\paragraph{Increased probability after state space truncation (Theorem~\ref{thm:ipt}). }
According to Theorem~\ref{thm:ipt}, the probability landscape
projected on the MEGs increase after state space truncation. We first
compute the steady state probability landscapes of $A$ obtained using
truncations at different sizes of the MEG ranging from $0$ to $119$
for $MEG_1$ while $MEG_2$ is fixed at 80 (Fig.~\ref{fig:tg1}C).  The
results are compared with the exact steady state landscape computed
using $MEG_1 = 120$ and $MEG_2 = 80$ (Fig.~\ref{fig:tg1}C, red line).
We then also similarly examine the steady state probability landscapes
of $B$ obtained using truncations at different sizes of $MEG_2$ from
$0$ to $79$ while $MEG_1$ is fixed at 120 (Fig.~\ref{fig:tg1}D).

Our results show clearly that all probabilities in the landscapes increase when more 
states are truncated at smaller MEG sizes (Fig.~\ref{fig:tg1}C and D).  The probability 
landscapes computed using larger sizes of MEGs (\textit{e.g.}, $MEG_1 = 50$, Fig.~\ref{fig:tg1}C, yellow line and 
$MEG_2 = 32$, Fig.~\ref{fig:tg1}D, yellow line)
are approaching the exact landscape (Fig.~\ref{fig:tg1}C and D, red line). 
The probability landscapes obtained using smaller MEGs deviate 
significantly from the exact probability landscape. The smaller the MEG size, the more 
significant the deviation is.  These numerical results are  completely 
consistent with Theorem~\ref{thm:ipt}.

\paragraph{Truncating additional MEGs does not decrease probabilities (Theorem~\ref{thm:Iij}). }
We further examine Theorem~\ref{thm:Iij}, \textit{i.e.}, the
probability landscape projected on one MEG increase with state space
truncation at another MEG.  We first compare the projected steady
state probability landscapes on $A$ obtained using truncations of
different sizes of MEGs ranging from $0$ to $80$ for $MEG_2$ while
$MEG_1$ is fixed at 120 (Fig.~\ref{fig:tg1}E) We compare the results
with the exact steady state landscape (Fig.~\ref{fig:tg1}E, red line).
We also similarly examine the projected steady state probability
landscapes of $B$ obtained using truncations at different sizes of
$MEG_1$ ranging from $0$ to $120$ while $MEG_2$ is fixed at 80
(Fig.~\ref{fig:tg1}F).

Our results clearly show that all probabilities on the landscapes 
of $MEG_1$ ($MEG_2$) increase when the state space is 
truncated at $MEG_2$ ($MEG_1$) (Fig.~\ref{fig:tg1}E and F).  
The probability landscapes computed using larger sizes of MEGs 
(\textit{e.g.}, $MEG_2 = 32$ in Fig.~\ref{fig:tg1}E, yellow line and 
$MEG_1 = 50$ in Fig.~\ref{fig:tg1}F, yellow line) are approaching 
the exact landscape using $MEG_1 = 120$ and $MEG_2 = 80$ 
(Fig.~\ref{fig:tg1}E and F, red line). 
However, the probability landscapes using smaller MEGs significantly 
deviate from the exact probability landscape. The smaller the MEG size, the more 
significant the deviation is.  These numerical results are completely 
consistent with Theorem~\ref{thm:Iij}.

\subsection*{Phage Lambda Bistable Epigenetic Switch}
The bistable epigenetic switch for lysogenic maintenance and lytic induction in 
phage lambda is one of the well-parameterized realistic gene regulatory system. 
The efficiency and stability of the switch have been extensively 
studied~\cite{Arkin1998,Aurell2002PRE,Aurell2002PRL,ZhuJBCB2004,Zhu2004}. 
Here we characterize the truncation error to the dCME solutions of the 
reaction network adapted from Cao {\it et al.}~\cite{Cao2010}. 
The network consists of $11$ different species and $50$ different reactions. The detailed reaction 
schemes and rate constants are shown in Table~\ref{tab:phage1}.

\paragraph{Molecular equivalent group (MEG). }
The network can be partitioned into two MEGs. The MEG$_1$ consists of the dimer
of CI protein $CI2$ and all complexes of operator sites bounded with $CI2$.
The MEG$_2$ consists of the dimer of Cro protein $Cro2$ and all complexes of
operator sites bounded with $Cro2$. 

\paragraph{Asymptotic convergence of errors (Theorem~\ref{thm:ace}). }
To numerically demonstrate Theorem~\ref{thm:ace}, we compute the true
error of the steady state solution to the dCME using sufficiently
large sizes of $MEG_1 = 80$ and $MEG_2 = 38$, which gives negligible
truncation error, with infinitesimally small boundary probabilities
$6.96 \times 10^{-31}$ for $MEG_1$ and $3.95 \times 10^{-32}$ for
$MEG_2$.  Solution obtained using these MEGs is therefore considered
to be exact.  With this exact steady state probability landscape, the
true truncation error $\Err^{(N)}$ at smaller sizes of $MEG_1$ and
$MEG_2$ can both be computed using Eqn.~(\ref{eqn:trueerr})
(Fig.~\ref{fig:ph1}A and B, blue dashed lines and crosses).  The
corresponding boundary probabilities $\pi^{(\infty)}_N$ or computed
error are computed from the exact steady state probability landscape
for both $MEG_1$ (Fig.~\ref{fig:ph1}A, green dashed line and circles)
and $MEG_2$ (Fig.~\ref{fig:ph1}B, green dashed line and circles).

Consistent with the statement in Theorem~\ref{thm:ace}, 
our results show that the true error $\Err^{(N)}$ 
(Fig.~\ref{fig:ph1}A and B, blue dashed lines and crosses) is 
bounded by the computed boundary probability $\pi^{(\infty)}_N$ 
(Fig.~\ref{fig:ph1}A and B, green dashed lines and circles) 
when the size of the MEG is sufficiently large. 
The insets in Fig.~\ref{fig:ph1}A  and B
show the ratios of the true errors to the computed errors at different 
sizes of the MEG, and the grey straight lines mark the ratio one. 
The computed errors are larger than the true errors when the black line 
is below the grey straight line (Fig.~\ref{fig:ph1}A and B, insets).
In this example, 
the computed boundary probability is greater than the 
true error when $MEG_1 \geq 24$ and $MEG_2 \geq 3$, 
as would be expected from Theorem~\ref{thm:ace}.

\paragraph{\textit{A priori} estimated error bound. }
To examine \textit{a priori} estimated upper bounds for 
the truncation errors in MEG$_1$ and MEG$_2$, 
we follow Eqn.~(\ref{eqn:alu}) and (\ref{eqn:blu}) to assign values of 
$\overline{\alpha}_{i}=k_{s1CI_{2}}$ and
$\underline{\beta}_{(i+1)}=[(i+1)-3]\cdot k_{dCI_{2}}$ for the $MEG_1$,
where the subscript $(i+1)$ is the total copy number of species $CI2$ in
the system. The subtraction of $3$ is necessary because up to $3$
copies of $CI2$ can be protected from degradation by binding to operator sites $OR1$, $OR2$, and $OR3$.
Similarly, we have
$\overline{\alpha}_{i}=k_{sCro_2}$ and $\underline{\beta}_{i+1}=[(i+1)-3]
\cdot k_{dCro_2}$ following Eqn.~(\ref{eqn:alu}) and (\ref{eqn:blu}) for the
$MEG_2$. 
We compute the \textit{a priori} estimated upper bounds of errors 
for different truncations of MEG$_1$ and MEG$_2$ using Eqn.~(\ref{eqn:upb})
(Fig.~\ref{fig:ph1}A and B, red solid lines).  
The true truncation errors and the \textit{a priori} estimated error bounds of MEG$_1$ and MEG$_2$ all 
decrease monotonically with increasing MEG sizes (Fig.~\ref{fig:ph1}A and B).  
The computed errors also monotonically decrease when the MEG sizes are larger 
than $13$ for MEG$_1$ and $4$ for MEG$_2$.  For both MEGs,
the {\it a priori}\/ estimated error bounds are larger
than computed errors at all MEG sizes.  They are also larger than the true
errors when the MEG sizes are sufficiently large.

\paragraph{Increased probability after state space truncation (Theorem~\ref{thm:ipt}). }
According to Theorem~\ref{thm:ipt}, the probability landscape
projected on the MEGs increase after state space truncation. We first
compute the steady state probability landscapes of $CI2$ obtained by
truncating $MEG_1$ at different sizes ranging from $0$ to $80$
while $MEG_2$ is fixed at $38$ (Fig.~\ref{fig:ph1}C).  The
results are compared with the exact steady state landscape computed
using $MEG_1 = 80$ and $MEG_2 = 38$ (Fig.~\ref{fig:ph1}C, red line).
We then also similarly examine the steady state probability landscapes
of $Cro2$ obtained by truncating at different sizes of $MEG_2$ from
$0$ to $38$ while $MEG_1$ is fixed at $80$ (Fig.~\ref{fig:ph1}D).

Our results show clearly that all probabilities in the landscapes increase when more 
states are truncated at smaller MEG sizes (Fig.~\ref{fig:ph1}C and D).  The probability 
landscapes computed using larger sizes of MEGs (\textit{e.g.}, $MEG_1 = 30$, Fig.~\ref{fig:ph1}C, yellow line and 
$MEG_2 = 8$, Fig.~\ref{fig:ph1}D, yellow line)
are approaching the exact landscape (Fig.~\ref{fig:ph1}C and D, red line). 
The probability landscapes obtained using smaller MEGs deviate 
significantly from the exact probability landscape. The smaller the MEG size, the more 
significant the deviation is.  These numerical results are completely 
consistent with Theorem~\ref{thm:ipt}.

\paragraph{Truncating additional MEGs does not decrease probabilities (Theorem~\ref{thm:Iij}). }
We further examine Theorem~\ref{thm:Iij}, \textit{i.e.}, the
probability landscape projected on one MEG increase with state space
truncation at another MEG.  We first compare the projected steady
state probability landscapes on $CI2$ obtained by truncating $MEG_2$ at 
different sizes ranging from $0$ to $38$ while
$MEG_1$ is fixed at $80$ (Fig.~\ref{fig:ph1}E).  We compare the results
with the exact steady state landscape (Fig.~\ref{fig:ph1}E, red line).
We also similarly examine the projected steady state probability
landscapes of $Cro2$ obtained by truncating at different sizes of
$MEG_1$ ranging from $0$ to $80$ while $MEG_2$ is fixed at $38$
(Fig.~\ref{fig:ph1}F).

Our results show that all probabilities on the landscapes 
of $MEG_1$ ($MEG_2$) increase when the state space is 
truncated at $MEG_2$ ($MEG_1$) (Fig.~\ref{fig:ph1}E and F).  
The probability landscapes computed using larger sizes of MEGs 
(\textit{e.g.}, $MEG_2 = 8$ in Fig.~\ref{fig:ph1}E, yellow line and 
$MEG_1 = 30$ in Fig.~\ref{fig:ph1}F, yellow line) are approaching 
the exact landscape using $MEG_1 = 80$ and $MEG_2 = 38$ 
(Fig.~\ref{fig:ph1}E and F, red line). 
However, the probability landscapes using smaller MEGs significantly 
deviate from the exact probability landscape. The smaller the MEG size, the more 
significant the deviation is.  These numerical results are completely 
consistent with Theorem~\ref{thm:Iij}.

\section*{Discussions and Conclusions}

Solving the discrete chemical master equation (dCME) is of fundamental
importance for studying stochasticity in reaction networks.  The main
challenges are the discrete nature of the states and the difficulty in
enumerating these states, as the size of the state space expands
rapidly when the network becomes more complex.  In this study, we
describe a novel approach for state space truncation.  Instead of
taking a high dimensional hypercube as the truncated state space, we
introduce the concept of molecular equivalence group (MEG), and
truncate the state space into the same or lower dimensional simplexes,
with the same effective copy number of molecules in each dimension by
taking advantage of the principle of mass conservation.  For complex
networks, the reduction of the size of the state space can be
dramatic.

Our study addresses a key issue in obtaining direct solution to
the dCME.  As state space truncation is inevitable, it is important to
quantify the errors of such truncations, so the accuracy of the dCME 
solutions can be assessed and
managed.  We have developed a general theoretical framework for
quantifying the errors of state space truncation on the steady state
probability landscape.  By decomposing the reaction network into MEGs,
the error contribution from each individual MEG is quantified.  This
critically important task is made possible through analyzing the
states on the reflecting boundary and their associated steady state
probabilities. The boundary probability analysis has been based on 
the construction of an aggregated continuous-time
Markov process by factoring the state space according to the total
numbers of molecules in each MEG.  With explicit formulas for
calculating conservative error bounds for the steady state, one can
easily calculate the {\it a priori}\/ error bounds for any given size 
of a MEG.  Furthermore, our theory allows the determination of
the minimally required sizes of MEGs if a predefined error
tolerance is to be satisfied.  
As shown in the examples, to determine the appropriate MEG sizes {\it
  a priori}, one can first calculate the estimated errors at different
sizes of each MEG, and  choose the minimal MEG sizes that
satisfies the overall error tolerance.
This eliminates the need of multiple
iterations of costly trial computations to solve the dCME for
determining the appropriate total copy numbers necessary to ensure
small truncation errors.  This is advantageous over
conventional numerical techniques, where errors are typically assessed
through post processing of trial solutions.

In complex networks, state truncation in one molecular group may
affect the errors of other molecular groups. By partitioning the
network into separate molecular equivalent groups (MEGs), the mutual
influence of the effects of state truncations in different groups can
be reduced.  In such cases, we have proved that the asymptotic errors in any
truncated MEG will not be under-estimated by the state truncations in
other MEGs.  Based on this conclusion, one can increase the size of
each particular MEG in order to achieve a small truncation error of
that MEG.  When the truncation error for every MEG is below the
prescribed threshold of error tolerance, the total truncation error of
the whole state space will be guaranteed to be bounded by the sum of
individual truncation errors in each MEG.

While our method ensures that there is no mass exchange between
different MEGs and often couplings between MEGs are weak, it does not
rule out the existence of possible strong couplings among MEGs. In the
example of the single gene expression model, there is a strong
coupling between mass-isolated mRNA MEG and the protein MEG. In this
case, protein synthesis strongly depends on the amount of available
mRNA. As a result, the protein probability distribution can be heavily
influenced by the choices of the mRNA MEG size, and its peak is
shifted when the size of mRNA MEG is near exhaustion (data not
shown). This issue rapidly disappears when MEG sizes become
sufficiently large to ensure that the truncation error to be smaller
than the specified error tolerance  (Fig.~\ref{fig:sge1}).

Our method differs from the finite state projection (FSP)
method~\cite{Munsky2006,Munsky2007}, which employs an absorbing
boundary state to calculate the truncation error.  Transitions from
any states in the available finite state space to any outside state
are send to the absorbing state, and the reactions are made
irreversible.  The truncation error in the FSP method is taken as the
probability mass on the absorbing boundary state.  It has two
components: one from the lost probability mass due to the state
truncation, the other from the trapped probability mass due to the
absorbing nature of the boundary state.  As time proceeds, the trapped
probability mass on the absorbing state will grow and dominate. At the
steady state, all probability mass will be trapped in the absorbing
state, which can no longer reflect the truncated probability mass.
Therefore, the FSP method cannot be used to study the long-term as well
as the steady state behavior of a stochastic network.

In contrast, our method employs a reflecting boundary and can
characterize the truncation errors in the steady state.  All
transitions between boundary and non-boundary states are retained
after state space truncation, and the reversible nature of 
transitions unaltered.  The reflecting boundaries allow analysis of
the steady state truncation error of each MEG.  Our method can be used
to study the steady state probability landscape.  Furthermore, our
method also allows direct computation of the distribution of first
passage time, an important problem in studying rare events in
biological networks currently relies heavily on sampling techniques.

We have also provided computational results of four stochastic networks,
namely, the birth-death process consisting of one MEG, the single gene
expression model, the genetic toggle switch model, and the phage lambda
epigenetic switch model, each consisting of two MEGs, respectively.  By
comparing true errors, computed errors, and \textit{a priori} estimated errors
at different truncation sizes, we have numerically verified the theorems
presented in this study: First, the true error for truncating a MEG is bounded
by the total probability mass on the reflecting boundary of the MEG
(Theorem~\ref{thm:ace}).  Second, the projected probability on one MEG
increases upon the state space truncation at this MEG (Theorem~\ref{thm:ipt}).
Third, the projected probability on one MEG  also increases when the state
space is truncated at another MEG (Theorem~\ref{thm:Iij}).  Furthermore, we
show that the \textit{a priori} estimated error bound are effective when the
network is truncated at a sufficiently large size of MEG.

Recent studies based on tensor representation of the transition rate matrices
show that the storage requirement of solving CME can be significantly reduced
and computational time improved~\cite{Kazeev2014,Liao2015}.  However, accurate
tensor representation and tensor-based approximation strongly depend on the
separability of system states, that is, whether the system can be decomposed
into a number of relatively independent smaller
sub-systems~\cite{Verstraete2006,Kazeev2014}.  While complete separability can
be achieved in some cases, {\it e.g.} the one-dimensional quantum spin
system~\cite{Verstraete2006}, errors are generally unknown for biological
networks that are not fully separable.

The tensor method of Liao {\it et al} can reduce the state space dramatically
for a number of networks~\cite{Liao2015}.  For example, the size of the state
space of the Fokker-Planck equation of the Schl\"ogl model is reduced from
$2.74 \times 10^{11}$ to $4.01\times 10^3 + 2.07 \times 10^5$, with a reduction
factor of $10^6$.  It will be interesting to further assess the reduction
factor if the full discrete CMEs instead of the Fokker-Planck equations of these
network models are solved so a direct comparison can be carried out.

Our finite buffer approach compares favorably with the tensor train
method of~\cite{Kazeev2014} for the network of enzymatic futile
cycles~\cite{Cao2013JCP}. This network is a closed system and
technically no finite buffer is required when the enumerated states
can fit into the computer memory, therefore analysis of truncation
error would be unnecessary.  Regardless, our approach of state
enumeration leads to a state space of only $1,071$ microstates, a
reflection of the $O(n!)$ order of reduction.  In contrast, the tensor
train method is based on a state space of a size of $2^{22} = 4.19
\times 10^6$.  Using our finite buffer method, both the time-evolving
and the steady state probability landscapes can be computed
efficiently in $<10$ seconds (data not shown), but the tensor-train
method requires $1.52 \times 10^4$ seconds for the time evolution of
$t=1$ to be computed as reported in~\cite{Kazeev2014}.  For the model
of toggle switch, computing the time-evolution of the probability
landscape up to $t=30$ seconds requires $14,541$ seconds or 4 hours of
wall clock time using the tensor-train method~\cite{Kazeev2014}.  
Our method completes the computation of the steady state probability 
landscape in {\it ca.} $3,300$ seconds or $55$ minutes of wall clock time.

We further note that our work complements tensor-based
methods~\cite{Kazeev2014,Liao2015}. Tensor-based methods directly
reduce the storage of the transition rate matrices~\cite{Kazeev2014},
without altering the hypercubic nature of the underlying state space.
In contrast, our method first reduces the state space by a factor of
$O(n!)$, leading to a dramatically reduced transition rate matrix.  It
is possible that there exist alternative approaches to construct
tensors of the transition rate matrix without assuming that the
truncated state space is a hypercube as is the case
in~\cite{Kazeev2014}. Whether our approach can be useful for further
reduction of storage and computational speed-up is a possible
direction for future exploration.

Overall, we have introduced an efficient method for state space
truncation and have developed theory to quantify the errors of state
space truncations. Results presented here provide a general framework
for high precision numerical solutions to a dCME.  It is envisioned
that the approach of direct solution of a dCME can be broadly applied
to many stochastic reaction networks, such as those found in systems
biology and in synthetic biology.

\section*{ACKNOWLEDGMENTS}
This work is supported by NIH grant GM079804, NSF grant MCB1415589,
and the Chicago Biomedical Consortium with support from the Searle
Funds at The Chicago Community Trust. We thank Dr. Ao Ma for helpful 
discussions and comments. YC is also supported by the LDRD program of CNLS at LANL.

\section*{APPENDIX}

\subsection*{\textbf{Proof of Lemma~\ref{lm:rma}}}

\begin{proof}
By sorting the state space according to the partition $\tilde{\Omega}^{(\infty)}$ 
and re-constructing the transition rate matrix $\tilde{\bA}$ in Eqn.~(\ref{eqn:Aaggreg1}),
the dCME can be re-written
as $\frac{d\tilde{\bp}^{(\infty)}(t)}{dt} = \tilde{\bA}
\tilde{\bp}^{(\infty)}(t)$, where $\tilde{\bp}^{(\infty)}$ is the
probability distribution on the partitioned state space. We sum up the
master equations over all microstates in each group $\mathcal{G}_i$
and obtain a separate aggregated equation for each group. As the re-ordered
matrix $\tilde{\bA}$ is a block tri-diagonal matrix, the summed discrete
chemical master equation is reduced to:
\begin{equation}
\begin{split}
\frac{d p^{(\infty)}(\mathcal{G}_0, t)}{dt} &= \frac{d\sum_{\bx \in \mathcal{G}_0} p(\bx,t)}{dt} = \left( \mathbbm{1}^T \bA_{0,0} \right) \tilde{\bp}^{(\infty)}(\mathcal{G}_{0},t) + \left( \mathbbm{1}^T \bA_{0,1} \right) \tilde{\bp}^{(\infty)}(\mathcal{G}_{1},t), \\
\frac{d p^{(\infty)}(\mathcal{G}_i, t)}{dt} &= \frac{d\sum_{\bx \in
\mathcal{G}_i} p(\bx,t)}{dt} = \left( \mathbbm{1}^T \bA_{i,i-1} \right) \tilde{\bp}^{(\infty)}(\mathcal{G}_{i-1},t) + \left( \mathbbm{1}^T \bA_{i,i} \right) \tilde{\bp}^{(\infty)}(\mathcal{G}_{i},t) + \left( \mathbbm{1}^T \bA_{i,i+1} \right) \tilde{\bp}^{(\infty)}(\mathcal{G}_{i+1},t), \\
\text{for } i=1, \cdots, \infty. \\
\end{split}
\label{eqn:cmeagg}
\end{equation}

The overall probability change of each group $\mathcal{G}_i$ depends
on the probability vector $\tilde{\bp}^{(\infty)}(\mathcal{G}_{i},t)$ itself,
as well as the probability vector $\tilde{\bp}^{(\infty)}(\mathcal{G}_{i-1},t)$ 
and the probability vector $\tilde{\bp}^{(\infty)}(\mathcal{G}_{i+1},t)$ 
of the immediate neighboring groups.
It also depends on the rates of synthesis and degradation reactions in
elements of $\bA_{i,i-1}$ and $\bA_{i,i+1}$, respectively, as well as rates of
coupling reactions in $\bA_{i,i}$. From the definition of transition
rate matrix given in Eqn.~(\ref{eqn:matA}), we have:
\begin{equation}
\begin{split}
\mathbbm{1}^T \bA_{0,0} &= - \mathbbm{1}^T \bA_{1,0},\\
\mathbbm{1}^T \bA_{i-1,i}  + \mathbbm{1}^T \bA_{i,i} &= -\mathbbm{1}^T \bA_{i+1,i}, 
\quad
\text{for } i=1, \cdots, \infty. \\
\end{split}
\label{eqn:stOnPartition}
\end{equation}

At the steady state when all $\frac{d p^{(\infty)}(\mathcal{G}_i)}{dt}
= 0$, we combine line 1 of Eqn.~(\ref{eqn:cmeagg}) and line 1 of
Eqn.~(\ref{eqn:stOnPartition}), and obtain:
$$
\left( \mathbbm{1}^T \bA_{1,0} \right) \tilde{\bpi}^{(\infty)}(\mathcal{G}_{0}) = \left( \mathbbm{1}^T \bA_{0,1} \right) \tilde{\bpi}^{(\infty)}(\mathcal{G}_{1}).
$$
From line 2 of Eqn.~(\ref{eqn:cmeagg})
 at steady state and after incorporating 
line 1 of Eqn.~(\ref{eqn:stOnPartition}), we have:
$
\left( \mathbbm{1}^T \bA_{1,2} \right) \tilde{\bpi}^{(\infty)}(\mathcal{G}_{2}) = 
  \left( \mathbbm{1}^T \bA_{0,0} \right) \tilde{\bpi}^{(\infty)}(\mathcal{G}_{0})
      -
  \left( \mathbbm{1}^T \bA_{1,1} \right) \tilde{\bpi}^{(\infty)}(\mathcal{G}_{1}).
$
After further incorporating  line 1 of Eqn.~(\ref{eqn:cmeagg})  
at steady state, we have
$ 
\left( \mathbbm{1}^T \bA_{1,2} \right) \tilde{\bpi}^{(\infty)}(\mathcal{G}_{2}) = 
   -  \left( \mathbbm{1}^T \bA_{0,1} \right) \tilde{\bpi}^{(\infty)}(\mathcal{G}_{1})
   -  \left( \mathbbm{1}^T \bA_{1,1} \right) \tilde{\bpi}^{(\infty)}(\mathcal{G}_{1}).
$
Incorporating line 2 of Eqn.~(\ref{eqn:stOnPartition}), we have:
$$
\left( \mathbbm{1}^T \bA_{2,1} \right) \tilde{\bpi}^{(\infty)}(\mathcal{G}_{1}) = 
   \left( \mathbbm{1}^T \bA_{1,2} \right) \tilde{\bpi}^{(\infty)}(\mathcal{G}_{2}).
$$
Assume
$
\left( \mathbbm{1}^T \bA_{i,\,i-1} \right) \tilde{\bpi}^{(\infty)}(\mathcal{G}_{i-1}) = 
   \left( \mathbbm{1}^T \bA_{i-1,\,i} \right) \tilde{\bpi}^{(\infty)}(\mathcal{G}_{i}),
$
we have from the $i$-the line of Eqn.~(\ref{eqn:cmeagg})
at the steady state
\begin{equation}
\begin{split}
\left( \mathbbm{1}^T \bA_{i,\,i+1} \right) \tilde{\bpi}^{(\infty)}(\mathcal{G}_{i+1}) 
&= 
  - \left( \mathbbm{1}^T \bA_{i,\,i-1} \right) \tilde{\bpi}^{(\infty)}(\mathcal{G}_{i-1})
  - \left( \mathbbm{1}^T \bA_{i,\,i} \right) \tilde{\bpi}^{(\infty)}(\mathcal{G}_{i})\\
&= 
  - \left( \mathbbm{1}^T \bA_{i-1,\,i} \right) \tilde{\bpi}^{(\infty)}(\mathcal{G}_{i})
  - \left( \mathbbm{1}^T \bA_{i,\,i} \right) \tilde{\bpi}^{(\infty)}(\mathcal{G}_{i})
.
\end{split}
\end{equation}
With the $i$-th line of Eqn.~(\ref{eqn:stOnPartition}), we further have:
$$
\left( \mathbbm{1}^T \bA_{i,\,i+1} \right) \tilde{\bpi}^{(\infty)}(\mathcal{G}_{i+1}) 
=
\left( \mathbbm{1}^T \bA_{i+1,\,i} \right) \tilde{\bpi}^{(\infty)}(\mathcal{G}_{i}). 
$$
Overall, we have:
\begin{equation}
\begin{split}
\left( \mathbbm{1}^T \bA_{1,0} \right) \tilde{\bpi}^{(\infty)}(\mathcal{G}_{0}) &= \left( \mathbbm{1}^T \bA_{0,1} \right) \tilde{\bpi}^{(\infty)}(\mathcal{G}_{1}), \\
\left( \mathbbm{1}^T \bA_{i+1,i} \right) \tilde{\bpi}^{(\infty)}(\mathcal{G}_{i}) &= \left( \mathbbm{1}^T \bA_{i,i+1} \right) \tilde{\bpi}^{(\infty)}(\mathcal{G}_{i+1}), \\
\text{for } i=1, \cdots, \infty. \\
\end{split}
\label{eqn:ss1}
\end{equation}
As both sides are constants, we can find $\alpha_i$ and $\beta_{i+1}$ such that:
\begin{equation}
\begin{split}
\left( \mathbbm{1}^T \bA_{i+1,i} \right) \tilde{\bpi}^{(\infty)}(\mathcal{G}_{i}) = \mathbbm{1}^T \alpha_{i} \tilde{\bpi}^{(\infty)}(\mathcal{G}_{i}) &= \alpha_{i} \mathbbm{1}^T \tilde{\bpi}^{(\infty)}(\mathcal{G}_{i}),\\
\left( \mathbbm{1}^T \bA_{i,i+1} \right) \tilde{\bpi}^{(\infty)}(\mathcal{G}_{i+1}) = \mathbbm{1}^T \beta_{i+1} \tilde{\bpi}^{(\infty)}(\mathcal{G}_{i+1}) &= \beta_{i+1} \mathbbm{1}^T \tilde{\bpi}^{(\infty)}(\mathcal{G}_{i+1}), 
\end{split}
\label{eqn:abdef1}
\end{equation}
for all $i = 0, 1, \cdots$, where $i$ is the total copy number of the MEG. 
We obviously have:
$$
\alpha_i = \left( \mathbbm{1}^T \bA_{i+1,i} \right) \cdot 
\frac{\tilde{\bpi}^{(\infty)}(\mathcal{G}_{i})}{\mathbbm{1}^T \tilde{\bpi}^{(\infty)}(\mathcal{G}_{i})} \quad \text{and} \quad
\beta_{i+1} = \left( \mathbbm{1}^T \bA_{i,i+1} \right) \cdot \frac{\tilde{\bpi}^{(\infty)}(\mathcal{G}_{i+1})}{\mathbbm{1}^T \tilde{\bpi}^{(\infty)}(\mathcal{G}_{i+1})}, 
$$
where $\alpha_i$ is the sum of column-sums of sub-matrix
$\bA_{i+1,i}$ weighted by the steady state probability distribution
$\tilde{\bpi}^{(\infty)}$ on group $\mathcal{G}_i$,
$\beta_{i+1}$ is the sum of column-summation of sub-matrix $\bA_{i,i+1}$
weighted by the steady state probability distribution on group
$\mathcal{G}_{i+1}$.

As $\mathbbm{1}^T \tilde{\bpi}^{(\infty)}(\mathcal{G}_{i})$ is the total steady
state probability mass over states in group $\mathcal{G}_i$, we
substitute Eqn.~(\ref{eqn:abdef1}) back into Eqn.~(\ref{eqn:ss1}) and
obtain the following relationship of steady state distribution on the
partitions of $\tilde{\Omega}^{\infty}$:
\begin{equation}
\begin{split}
\alpha_{0} \mathbbm{1}^T \tilde{\bpi}^{(\infty)}(\mathcal{G}_{0}) &= \beta_{1} \mathbbm{1}^T \tilde{\bpi}^{(\infty)}(\mathcal{G}_{1}), \\
\alpha_{i} \mathbbm{1}^T \tilde{\bpi}^{(\infty)}(\mathcal{G}_{i}) &= \beta_{i+1} \mathbbm{1}^T \tilde{\bpi}^{(\infty)}(\mathcal{G}_{i+1}), \\
\text{for } i=1, \cdots, \infty. \\
\end{split}
\label{eqn:abe}
\end{equation}
The steady state solution to Eqn.~(\ref{eqn:abe}) is equivalent to the
steady state solution of a dCME with the transition
rate matrix $\bB$ defined as in Eqn.~(\ref{eqn:bdmatinf}). 

\end{proof}

\subsection*{\textbf{Proof of Lemma~\ref{lm:fbs}}}

\begin{proof}
If $\lim_{N \rightarrow \infty} \sup\limits_{i > N}
\frac{\alpha^{(\infty)}_{i}}{\beta^{(\infty)}_{i+1}} \geq 1$ held,
then there would be an infinite number of terms
$\frac{\alpha^{(\infty)}_{i}}{\beta^{(\infty)}_{i+1}} > 1$.  There
should exist an integer $N'$ such that for all $i > N'$, we have
$\beta^{(\infty)}_{i+1} \leq \alpha^{(\infty)}_{i}$.  According to
Eqn.~(\ref{eqn:abe}), we would have 
$\tilde{\pi}^{(\infty)}_{i+1} \geq \tilde{\pi}^{(\infty)}_{i}$ 
in the steady state for all $i > N'$.
This contradicts with the assumption of a finite system, as the total
probability mass on boundary states increases monotonically as the net
molecular copy number of the network increases after $N'$.  This makes
the overall system a pure-birth process.  Therefore, for a finite
biological system, we have Eqn.~(\ref{eqn:fbs}).
\end{proof}

\subsection*{\textbf{Proof of Theorem~\ref{thm:ace}}}

\begin{proof}
From Eqn.~(\ref{eqn:pininf}), we can first derive an explicit expression of
the true error $\Err^{(N)}$ using the aggregated synthesis and
degradation rates $\alpha^{(\infty)}_{k}$ and $\beta^{(\infty)}_{k+1}$
given in Eqn.~(\ref{eqn:abdef}): 
\begin{equation}
\begin{aligned}
\Err^{(N)} &= 1 - \sum_{\bx \in \Omega^{(N)}} \pi^{(\infty)}(\bx) 
= 1 - \sum_{i = 0}^{N} \mathbbm{1}^T \tilde{\bpi}^{(\infty)}(\mathcal{G}_i) 
= 1 - \tilde{\pi}_0^{(\infty)} (1+\sum_{j=1}^{N} \prod_{k=0}^{j-1} \frac{\alpha^{(\infty)}_{k}}{\beta^{(\infty)}_{k+1}}) \\ 
&= 1 - \frac{1+\sum_{j=1}^{N} \prod_{k=0}^{j-1} \frac{\alpha^{(\infty)}_{k}}{\beta^{(\infty)}_{k+1}}}{1+\sum_{j=1}^{\infty} \prod_{k=0}^{j-1} \frac{\alpha^{(\infty)}_{k}}{\beta^{(\infty)}_{k+1}}} 
= \frac{\sum_{j=N+1}^{\infty} \prod_{k=0}^{j-1} \frac{\alpha^{(\infty)}_{k}}{\beta^{(\infty)}_{k+1}}}{1+\sum_{j=1}^{\infty} \prod_{k=0}^{j-1} \frac{\alpha^{(\infty)}_{k}}{\beta^{(\infty)}_{k+1}}} 
\label{eqn:err1}
\end{aligned}
\end{equation}
From Eqn.~(\ref{eqn:err1}), Eqn.~(\ref{eqn:pininf}), and Lemma~\ref{lm:fbs}, we have: 
\begin{equation}
\begin{aligned}
\frac{\Err^{(N)}}{\tilde{\pi}^{(\infty)}_N} &=
\frac{\sum_{j=N+1}^{\infty} \prod_{k=0}^{j-1}
  \frac{\alpha^{(\infty)}_{k}}{\beta^{(\infty)}_{k+1}}}{\prod_{k=0}^{N-1}
  \frac{\alpha^{(\infty)}_{k}}{\beta^{(\infty)}_{k+1}}} =
\frac{(\prod_{k=0}^{N-1}
  \frac{\alpha^{(\infty)}_{k}}{\beta^{(\infty)}_{k+1}})
  (\sum_{j=N+1}^{\infty} \prod_{k=N}^{j-1}
  \frac{\alpha^{(\infty)}_{k}}{\beta^{(\infty)}_{k+1}})}{\prod_{k=0}^{N-1}
  \frac{\alpha^{(\infty)}_{k}}{\beta^{(
\infty)}_{k+1}}} \\
&= \sum_{j=N+1}^{\infty} \prod_{k=N}^{j-1}
  \frac{\alpha^{(\infty)}_{k}}{\beta^{(\infty)}_{k+1}} 
  \leq
  \sum_{j=N+1}^{\infty}
   \left[ \sup\limits_{k \geq N} \left\{ \frac{\alpha^{(\infty)}_{k}}{\beta^{(\infty)}_{k+1}} \right\} \right]^{j-N} 
= \sum_{j=1}^{\infty}
   \left[ 
     \sup\limits_{k \geq N} \left\{ \frac{\alpha^{(\infty)}_{k}}{\beta^{(\infty)}_{k+1}} \right\} \right]^{j},
\end{aligned}
\end{equation}
When $N$ is sufficiently large, 
$ \sup\limits_{k \geq N} \left\{ \frac{\alpha^{(\infty)}_{k}}{\beta^{(\infty)}_{k+1}} \right\} < 1$ 
from Lemma~\ref{lm:fbs}, the terms in the infinite series 
$\sum_{j=1}^{\infty}
   \left[ \sup\limits_{k \geq N} \left\{ \frac{\alpha^{(\infty)}_{k}}{\beta^{(\infty)}_{k+1}} \right\} \right]^{j}
$
then forms a converging geometric series. Therefore, we have
$$
\sum_{j=1}^{\infty}
   \left[ 
          \sup\limits_{k \geq N} \left\{ \frac{\alpha^{(\infty)}_{k}}{\beta^{(\infty)}_{k+1}} \right\} 
   \right]^{j}
= \frac{
        \sup\limits_{k \geq N} \left\{ \frac{\alpha^{(\infty)}_{k}}{\beta^{(\infty)}_{k+1}} \right\} 
       }
       {1-  \sup\limits_{k \geq N} \left\{ \frac{\alpha^{(\infty)}_{k}}{\beta^{(\infty)}_{k+1}} \right\} },
$$
and the following inequality holds:
$$
\lim_{N \rightarrow \infty}  \frac{\Err^{(N)}}{\bar{\pi}^{(\infty)}_N}
\leq 
\lim_{N \rightarrow \infty} 
  \frac{
        \sup\limits_{k \geq N} \left\{ \frac{\alpha^{(\infty)}_{k}}{\beta^{(\infty)}_{k+1}} \right\} 
       }
       {1-  \sup\limits_{k \geq N} \left\{ \frac{\alpha^{(\infty)}_{k}}{\beta^{(\infty)}_{k+1}} \right\} }.
$$

Let $M \in \{N, \cdots, \infty \} $ be the integer such that
$ \frac{\alpha^{(\infty)}_{M}}
     {\beta^{(\infty)}_{M+1}}
=\mathop {\sup }\limits_{k \ge N}
\left\{
   \frac{\alpha^{(\infty)}_{k}}
        {\beta^{(\infty)}_{k+1}}
\right\}$,
we have the following inequality equivalent to Inequality (\ref{eqn:ace}): 
$$
\lim_{N \rightarrow \infty} 
\frac{\Err^{(N)}}{\bar{\pi}^{(\infty)}_N}\leq \lim_{N \rightarrow \infty} \frac{\frac{\alpha^{(\infty)}_{M}}{\beta^{(\infty)}_{M+1}}}{1-\frac{\alpha^{(\infty)}_{M}}{\beta^{(\infty)}_{M+1}}}.
$$
\end{proof}

\subsection*{\textbf{Proof of Theorem~\ref{thm:ipt}}}

\begin{proof}
We first consider two truncated state spaces $\tilde{\Omega}^{(N)}$
and $\tilde{\Omega}^{(N+1)}$. Following Eqn.~(\ref{eqn:cmeagg}), two
finite sets of the block chemical master equation can be constructed
for these two state spaces.  The first set containing $N$ equations is
built on the state space $\tilde{\Omega}^{(N)}$.
\begin{equation}
\begin{split}
\frac{d p^{(N)}(\mathcal{G}_0, t)}{dt} &= \left( \mathbbm{1}^T \bA_{0,0} \right) \tilde{\bp}^{(N)}(\mathcal{G}_{0},t) + \left( \mathbbm{1}^T \bA_{0,1} \right) \tilde{\bp}^{(N)}(\mathcal{G}_{1},t), \\
\frac{d p^{(N)}(\mathcal{G}_{i}, t)}{dt} &= \left( \mathbbm{1}^T \bA_{i,i-1} \right) \tilde{\bp}^{(N)}(\mathcal{G}_{i-1},t) + \left( \mathbbm{1}^T \bA_{i,i} \right) \tilde{\bp}^{(N)}(\mathcal{G}_{i},t) + \left( \mathbbm{1}^T \bA_{i,i+1} \right) \tilde{\bp}^{(N)}(\mathcal{G}_{i+1},t), \\
\text{for } i=1, \cdots, N-1, \\
\frac{d p^{(N)}(\mathcal{G}_N, t)}{dt} &= \left( \mathbbm{1}^T \bA_{N,N-1} \right) \tilde{\bp}^{(N)}(\mathcal{G}_{N-1},t) + \left( \mathbbm{1}^T \bA_{N,N} \right) \tilde{\bp}^{(N)}(\mathcal{G}_{N},t). \\
\end{split}
\label{eqn:cmeaggn}
\end{equation}
The second set is built on the state space $\tilde{\Omega}^{(N+1)}$ 
containing $N+1$ equations. 
\begin{equation}
\begin{split}
\frac{d p^{(N+1)}(\mathcal{G}_0, t)}{dt} &= \left( \mathbbm{1}^T \bA_{0,0} \right) \tilde{\bp}^{(N+1)}(\mathcal{G}_{0},t) + \left( \mathbbm{1}^T \bA_{0,1} \right) \tilde{\bp}^{(N+1)}(\mathcal{G}_{1},t), \\
\frac{d p^{(N+1)}(\mathcal{G}_{i}, t)}{dt} &= \left( \mathbbm{1}^T \bA_{i,i-1} \right) \tilde{\bp}^{(N+1)}(\mathcal{G}_{i-1},t) + \left( \mathbbm{1}^T \bA_{i,i} \right) \tilde{\bp}^{(N+1)}(\mathcal{G}_{i},t) + \left( \mathbbm{1}^T \bA_{i,i+1} \right) \tilde{\bp}^{(N+1)}(\mathcal{G}_{i+1},t), \\
\text{for } i=1, \cdots, N-1, \\
\frac{d p^{(N+1)}(\mathcal{G}_{N}, t)}{dt} &= \left( \mathbbm{1}^T \bA_{N,N-1} \right) \tilde{\bp}^{(N+1)}(\mathcal{G}_{N-1},t) + \left( \mathbbm{1}^T \bA_{N,N} \right) \tilde{\bp}^{(N+1)}(\mathcal{G}_{N},t) + \left( \mathbbm{1}^T \bA_{N,N+1} \right) \tilde{\bp}^{(N+1)}(\mathcal{G}_{N+1},t), \\
\frac{d p^{(N+1)}(\mathcal{G}_{N+1}, t)}{dt} &= \left( \mathbbm{1}^T \bA_{N+1,N} \right) \tilde{\bp}^{(N+1)}(\mathcal{G}_{N},t) + \left( \mathbbm{1}^T \bA_{N+1,N+1} \right) \tilde{\bp}^{(N+1)}(\mathcal{G}_{N+1},t). \\
\end{split}
\label{eqn:cmeaggn1}
\end{equation}
At steady state, the left-hand side of the equations are zeros.  
For the first $N$ equations, the corresponding block matrices are the same for 
both state spaces $\tilde{\Omega}^{(N)}$ and $\tilde{\Omega}^{(N+1)}$. 
We can then subtract the right-hand side of 
Eqn.~(\ref{eqn:cmeaggn1}) from Eqn.~(\ref{eqn:cmeaggn}) and  obtain the following steady state equations: 
\begin{equation}
\label{eqn:CMEAgr1}
\begin{split}
 \mathbbm{1}^T \bA_{0,0} {\Delta\bpi}_{0} + \mathbbm{1}^T \bA_{0,1} {\Delta\bpi}_{1}=0, \\
 \mathbbm{1}^T \bA_{i,i-1} {\Delta\bpi}_{i-1} + \mathbbm{1}^T \bA_{i,i} {\Delta\bpi}_{i} + \mathbbm{1} \bA_{i,i+1} {\Delta\bpi}_{i+1}=0, \\
\text{for } i=1, \cdots, N-1, \\
\end{split}
\end{equation}
where ${\Delta\bpi}_{i}={\bpi}^{(N)}_i - {\bpi}^{(N+1)}_i$ is the
steady state probability difference between the state group $\mathcal{G}_i$ in the dCME 
on $\tilde{\Omega}^{(N)}$ and $\tilde{\Omega}^{(N+1)}$.
However, the block sub-matrix $\bA_{N,N}$ of the boundary group $\mathcal{G}_N$ is 
different between the two state spaces.  
From the construction of the 
aggregated dCME matrix $\tilde{\bA}$, 
columns of the full matrices $\tilde{\bA}^{(N+1)}$ over $\tilde{\Omega}^{(N+1)}$  
and $\tilde{\bA}^{N}$ over $\tilde{\Omega}^{N}$ all sum to 0 (see Eqn~\ref{eqn:stOnPartition}).    
We use $\bA^{(N)}_{i,\,j}$ to denote the block sub-matrix 
of the group $\mathcal{G}_N$ for the state space $\tilde{\Omega}^{(N)}$, and use 
$\bA^{(N+1)}_{i,\,j}$ to denote the corresponding block sub-matrix 
for the state space $\tilde{\Omega}^{(N+1)}$.  
From the $N$-th line of the truncated version of Eqn~({\ref{eqn:stOnPartition}}), we have 
$
\mathbbm{1}^T{\bA_{N-1,\,N}^{(N+1)}} + \mathbbm{1}^T{\bA_{N,\,N}^{(N+1)}} + 
                                         \mathbbm{1}^T{\bA_{N+1,\, N}^{(N+1)}} = 0
$
for $\tilde{\Omega}^{(N+1)}$
and
$
\mathbbm{1}^T{\bA_{N-1,\,N}^{(N)}} + \mathbbm{1}^T{\bA_{N,\,N}^{(N)}} = 0
$
for $\tilde{\Omega}^{(N)}$. 
Since ${\bA_{N-1,N}^{(N)}} = {\bA_{N-1,N}^{(N+1)}}$,  
we have the following property 
\begin{equation}
\mathbbm{1}^T{\bA_{N,N}^{(N+1)}}=\mathbbm{1}^T{\bA_{N,N}^{(N)}}-\mathbbm{1}^T{\bA_{N+1,N}^{(N+1)}},
\label{eqn:prpt1}
\end{equation}
We also have
\begin{equation}
\mathbbm{1}^T{\bA_{N+1,N+1}^{(N+1)}}=-\mathbbm{1}^T{\bA_{N,N+1}^{(N+1)}}. 
\label{eqn:prpt2}
\end{equation}

From Eqn.~(\ref{eqn:cmeaggn}), we have for the steady state the probability of the state group 
$\mathcal{G}_N$ over the state space $\tilde{\Omega}^{(N)}$ as:
\begin{equation}
\mathbbm{1}^T{\bA_{N,N-1}^{(N)}} \bpi_{N-1}^{(N)}  + \mathbbm{1}^T{\bA_{N,N}^{(N)}} \bpi_{N}^{(N)} = 0,
\label{eqn:nssn}
\end{equation}
From Eqn.~(\ref{eqn:cmeaggn1}), we  have for the steady state the
probability of the state group $\mathcal{G}_N$ and $\mathcal{G}_{N+1}$
over the state space $\tilde{\Omega}^{(N+1)}$ as:
\begin{equation}
\mathbbm{1}^T{\bA_{N,N-1}^{(N+1)}} \bpi_{N-1}^{(N+1)}  + \mathbbm{1}^T{\bA_{N,N}^{(N+1)}} \bpi_{N}^{(N+1)} + \mathbbm{1}^T{\bA_{N,N+1}^{(N+1)}} \bpi_{N+1}^{(N+1)} = 0, 
\label{eqn:n1ssn}
\end{equation}
and 
\begin{equation}
\mathbbm{1}^T{\bA_{N+1,N}^{(N+1)}} \bpi_{N}^{(N+1)}  + \mathbbm{1}^T{\bA_{N+1,N+1}^{(N+1)}} \bpi_{N+1}^{(N+1)} = 0,
\label{eqn:n1ssn1}
\end{equation}
respectively. 

As ${\bA_{N,N-1}^{(N+1)}} = {\bA_{N,N-1}^{(N)}}$,
we subtract Eqn.~(\ref{eqn:n1ssn}) from Eqn.~(\ref{eqn:nssn}), and obtain: 
$$
\mathbbm{1}^T{\bA_{N,N-1}} \Delta \bpi_{N-1} + \mathbbm{1}^T{\bA_{N,N}^{(N)}} \bpi_{N}^{(N)} - \mathbbm{1}^T{\bA_{N,N}^{(N+1)}} \bpi_{N}^{(N+1)} - \mathbbm{1}^T{\bA_{N,N+1}^{(N+1)}} \bpi_{N+1}^{(N+1)} = 0.
$$
It can be re-written by applying the matrix property of Eqn.~(\ref{eqn:prpt1}) as: 
$$
\mathbbm{1}^T{\bA_{N,N-1}} \Delta \bpi_{N-1} + \mathbbm{1}^T{\bA_{N,N}^{(N)}} \Delta \bpi_{N} + \mathbbm{1}^T{\bA_{N+1,N}^{(N+1)}} \bpi_{N}^{(N+1)} - \mathbbm{1}^T{\bA_{N,N+1}^{(N+1)}} \bpi_{N+1}^{(N+1)} = 0.
$$
By using the matrix property in Eqn.~(\ref{eqn:prpt2}), we can further re-write it as: 
$$
\mathbbm{1}^T{\bA_{N,N-1}} \Delta \bpi_{N-1} + \mathbbm{1}^T{\bA_{N,N}^{(N)}} \Delta \bpi_{N} + \mathbbm{1}^T{\bA_{N+1,N}^{(N+1)}} \bpi_{N}^{(N+1)} + \mathbbm{1}^T{\bA_{N+1,N+1}^{(N+1)}} \bpi_{N+1}^{(N+1)} = 0.
$$
From Eqn.~(\ref{eqn:n1ssn1}), the last two terms sum to 0. Therefore,
we obtain the ($N+1$)-st equation of the steady state probability difference as: 
$$
\mathbbm{1}^T{\bA_{N,N-1}} \Delta \bpi_{N-1} + \mathbbm{1}^T{\bA_{N,N}^{(N)}} \Delta \bpi_{N} = 0.
$$

Taken together, we have the set of equations for steady state probability differences for all $N+1$ blocks as: 
\begin{equation}
\begin{split}
 \mathbbm{1}^T \bA_{0,0} {\Delta\bpi}_{0} + \mathbbm{1}^T \bA_{0,1} {\Delta\bpi}_{1}=0, \\
 \mathbbm{1}^T \bA_{i,i-1} {\Delta\bpi}_{i-1} + \mathbbm{1}^T \bA_{i,i} {\Delta\bpi}_{i} + \mathbbm{1}^T \bA_{i,i+1} {\Delta\bpi}_{i+1}=0, \\
 \text{for } i=1, \cdots, N-1, \\
 \mathbbm{1}^T{\bA_{N,N-1}} \Delta \bpi_{N-1} + \mathbbm{1}^T{\bA_{N,N}} \Delta \bpi_{N} = 0,
\end{split}
\label{eqn:CMEAgr2}
\end{equation}
where all block sub-matrices are identical between those over the
state spaces $\tilde{\Omega}^{(N)}$ and $\tilde{\Omega}^{(N+1)}$. 
We therefore obtain the set of equations of differences in steady state probability equivalent to Eqn.~(\ref{eqn:ss1}):
\begin{equation}
\begin{split}
\mathbbm{1}^T \bA_{i,i-1} {\Delta\bpi}_{i-1} &= \mathbbm{1}^T \bA_{i-1,i} {\Delta\bpi}_{i}, 
\text{for } i=1, \cdots, N, \\
\end{split}
\end{equation}
which produces the same steady state solution as that of
Eqn.~(\ref{eqn:ss1}) after scaling by a constant. 
As probability vector solution to Eqn.~(\ref{eqn:ss1}) has non-negative elements,
this equivalence implies that all
elements in each ${\Delta \bpi}_{i}$ have the same sign. 
As the total steady state probability mass in both state spaces sum up
to $1$,
$$\sum_{i=1}^N \tilde{\pi}_i^{(N)} = \sum_{i=1}^{N+1} \tilde{\pi}_i^{(N+1)} = 1,$$
we therefore know that the total probability differences is non-negative: 
$$
\sum_{i = 1}^N \Delta \tilde{\pi}_i = \sum_{i = 1}^N \tilde{\pi}_i^{(N)} - \sum_{i = 1}^N \tilde{\pi}_i^{(N+1)} = 1 - (1 - \tilde{\pi}_{N + 1}^{(N+1)}) = \tilde{\pi}_{N + 1}^{(N+1)} \geq 0. 
$$ 
Therefore, the probability difference of each
individual $\mathcal{G}_i$ between two state spaces must be
non-negative: 
$$
\Delta \tilde{\pi}_i = \tilde{\pi}_i^{(N)}-\tilde{\pi}_i^{(N+1)} \geq
0, \quad i = 0, 1, \cdots, N. 
$$ 
This can be generalized. As $N$ increases to infinity,  
we have:
$$
\tilde{\pi}_i^{(N)} \geq \tilde{\pi}_i^{(N+1)} \geq \cdots \geq \tilde{\pi}_i^{(\infty)}, 
\quad i = 0, 1, \cdots, N. 
$$

\end{proof}

\subsection*{\textbf{Proof of Theorem~\ref{thm:Iij}}}

\begin{proof}
For convenience, we use $M = N_i$ to denote the maximum net copy number 
in the truncated  $i$-th MEG. 
We first aggregate the state space $\Omega^{(\mathcal{I}_j)}$ into infinitely many 
groups $\{ \mathcal{G}_0, \mathcal{G}_1, \cdots, \mathcal{G}_M, \mathcal{G}_{M+1}, \cdots \}$ 
according to the net copy number in the $i$-th MEG. 
We then re-construct the permuted matrix $\tilde{\bA}^{(\mathcal{I}_j)}$ 
according to this aggregation. 
We have: 
\begin{equation}
\tilde{\bA}^{(\mathcal{I}_j)} = \left( {\begin{array}{c|c}
   {\bA_{g,h}^{(\mathcal{I}_j)}} & {\bA_{g,l}^{(\mathcal{I}_j)}} \\
	\hline 
   {\bA_{k,h}^{(\mathcal{I}_j)}} & {\bA_{k,l}^{(\mathcal{I}_j)}} \\
\end{array}} \right) 
= 
\left( {\begin{array}{c|c}
   {\bA_{\textbf{1},\textbf{1}}^{(\mathcal{I}_j)}} & {\bA_{\textbf{1},\textbf{2}}^{(\mathcal{I}_j)}} \\
	 \hline
   {\bA_{\textbf{2},\textbf{1}}^{(\mathcal{I}_j)}} & {\bA_{\textbf{2},\textbf{2}}^{(\mathcal{I}_j)}} \\
\end{array}} \right), \text{for } 0 \leq g,h \leq M, \text{ and } k,l \geq M+1,
\label{eqn:Aaggi}
\end{equation}
where the subscripts $m$ and $n$ of each block matrix $\bA_{m,n}^{(\mathcal{I}_j)}$ 
indicate the actual net copy numbers of 
the corresponding aggregated states of the $i$-th MEG.
Next, 
we further partition the matrix into four blocks by 
truncating  the $i$-th MEG at the maximum copy number of $M$. 
Specifically, $\bA_{\textbf{1},\textbf{1}}^{(\mathcal{I}_j)}$ in the right-hand side  of
Eqn.~(\ref{eqn:Aaggi}) is the 
north-west corner sub-matrix of $\tilde{\bA}^{(\mathcal{I}_j)}$, 
which contains all transitions between microstates in the state space $\Omega^{(\mathcal{I}_{i,j})}$: 
\begin{equation}
\bA_{\textbf{1},\textbf{1}}^{(\mathcal{I}_j)} = \left( {\begin{array}{c}
   {\bA_{g,h}^{(\mathcal{I}_j)}} 
\end{array}} \right) = \{A_{\bx_m, \bx_n}\}, \; \bx_m, \, \bx_n \in \Omega^{(\mathcal{I}_{i,j})}, 
\text{ and } 0 \leq g,h \leq M. 
\label{eqn:Aaggi11}
\end{equation}
$\bA_{\textbf{1},\textbf{2}}^{(\mathcal{I}_j)}$ is the north-east corner sub-matrix of $\tilde{\bA}^{(\mathcal{I}_j)}$, 
which contains all transitions from microstates in state space 
$\Omega^{(\mathcal{I}_{j})} / \Omega^{(\mathcal{I}_{i,j})}$ to
microstates in state space $\Omega^{(\mathcal{I}_{i,j})}$: 
\begin{equation}
\bA_{\textbf{1},\textbf{2}}^{(\mathcal{I}_j)} = \left( {\begin{array}{c}
   {\bA_{g,l}^{(\mathcal{I}_j)}}
\end{array}} \right) = \{A_{\bx_m, \bx_n}\}, \; \bx_m \in \Omega^{(\mathcal{I}_{i,j})}, \, \bx_n \in \Omega^{(\mathcal{I}_{j})} / \Omega^{(\mathcal{I}_{i,j})}, 
\text{ and } 0 \leq g \leq M, l \geq M+1. 
\label{eqn:Aaggi12}
\end{equation}
$\bA_{\textbf{2},\textbf{1}}^{(\mathcal{I}_j)}$ is the south-west corner sub-matrix of $\tilde{\bA}^{(\mathcal{I}_j)}$, 
which contains all transitions from microstates in state space 
$\Omega^{(\mathcal{I}_{i,j})}$ to microstates in state space 
$\Omega^{(\mathcal{I}_{j})} / \Omega^{(\mathcal{I}_{i,j})}$: 
\begin{equation}
\bA_{\textbf{2},\textbf{1}}^{(\mathcal{I}_j)} = \left( {\begin{array}{c}
	 {\bA_{k,h}^{(\mathcal{I}_j)}} 
\end{array}} \right) = \{A_{\bx_m, \bx_n}\}, \; \bx_m \in \Omega^{(\mathcal{I}_{j})} / \Omega^{(\mathcal{I}_{i,j})}, \, \bx_n \in \Omega^{(\mathcal{I}_{i,j})}, 
\text{ and } 0 \leq h \leq M, k \geq M+1. 
\label{eqn:Aaggi21}
\end{equation}
and $\bA_{\textbf{2},\textbf{2}}^{(\mathcal{I}_j)}$ is the south-east corner sub-matrix of $\tilde{\bA}^{(\mathcal{I}_j)}$, 
which contains all transitions between microstates in state space 
$\Omega^{(\mathcal{I}_{j})} / \Omega^{(\mathcal{I}_{i,j})}$: 
\begin{equation}
\bA_{\textbf{2},\textbf{2}}^{(\mathcal{I}_j)} = \left( {\begin{array}{cc}
	 {\bA_{k,l}^{(\mathcal{I}_j)}} 
\end{array}} \right) = \{A_{\bx_m, \bx_n}\}, \quad \mbox{ with  }\bx_m \mbox{ and } \bx_n \in \Omega^{(\mathcal{I}_{j})} / \Omega^{(\mathcal{I}_{i,j})}, 
\text{ and } k,l \geq M+1. 
\label{eqn:Aaggi22}
\end{equation}

We now truncate the state space at the maximum copy number $M$ of the
$i$-th MEG.  A matrix $\bA^{(\mathcal{I}_{i,j})}$ on the truncated
state space $\Omega^{(\mathcal{I}_{i,j})}$ using the same partition
$\{ \mathcal{G}_0, \mathcal{G}_1, \cdots, \mathcal{G}_M \}$ can be
constructed as:
\begin{equation}
\tilde{\bA}^{(\mathcal{I}_{i,j})} = \left( {\begin{array}{c}
   {\bA_{g,h}^{(\mathcal{I}_{i,j})}} 
\end{array}} \right), 
\text{ and } 0 \leq g,h \leq M. 
\label{eqn:Aaggij}
\end{equation}

Similar to the matrix $\tilde{\bA}$ in Eqn.~(\ref{eqn:Aaggreg1}), 
both matrices $\tilde{\bA}^{(\mathcal{I}_{i,j})}$
and $\tilde{\bA}^{(\mathcal{I}_{j})}$ are tri-diagonal matrix with ${\bA_{m,n}^{(\mathcal{I}_{i,j})}}=0$ 
and ${\bA_{m,n}^{(\mathcal{I}_{j})}}=0$ for any $|m - n| > 1$. 

 Matrix $\tilde{\bA}^{(\mathcal{I}_{i,j})}$ and sub-matrix
  $\bA_{\textbf{1},\textbf{1}}^{(\mathcal{I}_j)}$ reside on the same
  state space $\Omega^{(\mathcal{I}_{i,j})}$ and have exactly the same
  permutation, \textit{i.e.}, the matrix element
  $A^{(\mathcal{I}_{i,j})}_{\bx_m,\bx_n} \in
  \tilde{\bA}^{(\mathcal{I}_{i,j})}$ and
  $A^{(\mathcal{I}_{j})}_{\bx_m,\bx_n} \in
  \bA_{\textbf{1},\textbf{1}}^{(\mathcal{I}_j)}$ describes the same
  transitions between microstates $\bx_m, \, \bx_n \in
  \Omega^{(\mathcal{I}_{i,j})} \subset \Omega^{(\mathcal{I}_{j})}$.
  Only diagonal elements in $ {\bA_{M,M}^{(\mathcal{I}_{i,j})}} $ have
  different rates.  By construction, $\mathcal{G}_M$ and
  $\mathcal{G}_{M+1}$ are the only two aggregated groups that are
  involved in transition between states across the boundary of
  $\Omega^{(\mathcal{I}_{i,j})}$.  The sub-matrix
  $\bA_{M+1,M}^{(\mathcal{I}_j)}$ is the only nonzero sub-matrix in
  $\bA_{\textbf{2},\textbf{1}}^{(\mathcal{I}_j)}$, which forms the
  reflection boundary and is involved in synthesis reactions from microstates
  in group $\mathcal{G}_M$ to microstates in $\mathcal{G}_{M+1}$.  As
  a property of the rate matrix, we have 
\begin{equation}
\mathbbm{1}^T \bA_{M-1,M}^{(\mathcal{I}_j)}
+\mathbbm{1}^T \bA_{M,M}^{(\mathcal{I}_j)}
+\mathbbm{1}^T \bA_{M+1,M}^{(\mathcal{I}_j)} = \textbf{0}^T, 
\label{eqn:z-M-1.M}
\end{equation}
and 
\begin{equation}
\mathbbm{1}^T \bA_{M-1,M}^{(\mathcal{I}_{i,j})}
+\mathbbm{1}^T \bA_{M,M}^{(\mathcal{I}_{i,j})} = \textbf{0}^T. 
\label{eqn:i-j-M}
\end{equation}
Since
$
\bA_{M-1,M}^{(\mathcal{I}_j)} = \bA_{M-1,M}^{(\mathcal{I}_{i,j})},
$
we have from Eqn~(\ref{eqn:i-j-M})
$
\mathbbm{1}^T \bA_{M,M}^{(\mathcal{I}_{i,j})} =
- \mathbbm{1}^T \bA_{M-1,M}^{(\mathcal{I}_{i,j})} = 
- \mathbbm{1}^T \bA_{M-1,M}^{(\mathcal{I}_{i,j})}. 
$
With Eqn~(\ref{eqn:z-M-1.M}), we further have
$$\mathbbm{1}^T \bA_{M,M}^{(\mathcal{I}_{i,j})} = 
 \mathbbm{1}^T \bA_{M,M}^{(\mathcal{I}_j)}
+\mathbbm{1}^T \bA_{M+1,M}^{(\mathcal{I}_j)}.$$
By  construction,  the only differences between 
the sub-matrix $\bA_{M,M}^{(\mathcal{I}_{i,j})}$ and 
$\bA_{M,M}^{(\mathcal{I}_j)}$ are in the diagonal elements. 
Therefore, we have 
$$\bA_{M,M}^{(\mathcal{I}_{i,j})} = 
 \bA_{M,M}^{(\mathcal{I}_j)} + \diag(\mathbbm{1}^T
 \bA_{M+1,M}^{(\mathcal{I}_j)}).$$
That is: $$\tilde{\bA}^{(\mathcal{I}_{i,j})} = 
\bA_{\textbf{1},\textbf{1}}^{(\mathcal{I}_j)} + 
\diag(\mathbbm{1}^T \bA_{\textbf{2},\textbf{1}}^{(\mathcal{I}_j)}).$$ 
For convenience, we use the notation $\bR_{\textbf{2},\textbf{1}}^{(\mathcal{I}_j)} = \diag(\mathbbm{1}^T \bA_{\textbf{2},\textbf{1}}^{(\mathcal{I}_j)})$, 
and have $\tilde{\bA}^{(\mathcal{I}_{i,j})} = 
\bA_{\textbf{1},\textbf{1}}^{(\mathcal{I}_j)} + \bR_{\textbf{2},\textbf{1}}^{(\mathcal{I}_j)}$. 
We partition the steady state vector $\bpi^{(\mathcal{I}_{j})}$ accordingly into 
two sub-vectors: $\bpi^{(\mathcal{I}_{j})} = 
(\bpi^{(\mathcal{I}_{j})}_{\textbf{1}},\, \bpi^{(\mathcal{I}_{j})}_{\textbf{2}})$, 
where $\bpi^{(\mathcal{I}_{j})}_{\textbf{1}}$ corresponds to states in $\Omega^{(\mathcal{I}_{i,j})}$, 
and $\bpi^{(\mathcal{I}_{j})}_{\textbf{2}}$ corresponds to states 
in $\Omega^{(\mathcal{I}_{j})} / \Omega^{(\mathcal{I}_{i,j})}$. 
As $\tilde{\bA}^{(\mathcal{I}_{j})} \bpi^{(\mathcal{I}_{j})} = \textbf{0}$, we have: 
$$
\bA^{(\mathcal{I}_{j})}_{\textbf{1},\textbf{1}} \bpi^{(\mathcal{I}_{j})}_{\textbf{1}} + 
\bA^{(\mathcal{I}_{j})}_{\textbf{1},\textbf{2}} \bpi^{(\mathcal{I}_{j})}_{\textbf{2}} = \textbf{0}, 
$$
therefore 
$$
\left[ \tilde{\bA}^{(\mathcal{I}_{i,j})} - \bR^{(\mathcal{I}_{j})}_{\textbf{2},\textbf{1}} \right] 
\bpi^{(\mathcal{I}_{j})}_{\textbf{1}} + \bA^{(\mathcal{I}_{j})}_{\textbf{1},\textbf{2}} 
\bpi^{(\mathcal{I}_{j})}_{\textbf{2}} = \textbf{0}. 
$$
Hence, we have:
\begin{equation}
\tilde{\bA}^{(\mathcal{I}_{i,j})} \bpi^{(\mathcal{I}_{j})}_{\textbf{1}} = 
\bR^{(\mathcal{I}_{j})}_{\textbf{2},\textbf{1}} \bpi^{(\mathcal{I}_{j})}_{\textbf{1}} - 
\bA^{(\mathcal{I}_{j})}_{\textbf{1},\textbf{2}} \bpi^{(\mathcal{I}_{j})}_{\textbf{2}}. 
\label{eqn:bndry}
\end{equation}
As all off-diagonal entries of transition rate matrix $\tilde{\bA}^{(\mathcal{I}_{j})}$ 
are non-negative, we know that $\bA_{\textbf{1},\textbf{2}}^{(\mathcal{I}_j)} \geq \textbf{0}$, 
and $\bR_{\textbf{2},\textbf{1}}^{(\mathcal{I}_j)} \geq \textbf{0}$. 
Since $\bpi^{(\mathcal{I}_{j})} = 
(\bpi^{(\mathcal{I}_{j})}_{\textbf{1}}, \bpi^{(\mathcal{I}_{j})}_{\textbf{2}})$ 
is the steady state distribution of the rate matrix $\tilde{\bA}^{(\mathcal{I}_j)}$ 
with $\bpi^{(\mathcal{I}_{j})}_{\textbf{1}} \geq \textbf{0}$ 
and $\bpi^{(\mathcal{I}_{j})}_{\textbf{2}} \geq \textbf{0}$, 
we have $\bA_{\textbf{1},\textbf{2}}^{(\mathcal{I}_j)} 
\bpi^{(\mathcal{I}_{j})}_{\textbf{2}} \geq \textbf{0}$, and 
$\bR_{\textbf{2},\textbf{1}}^{(\mathcal{I}_j)} \bpi^{(\mathcal{I}_{j})}_{\textbf{1}} \geq \textbf{0}$. 
As all columns of matrix $\tilde{\bA}^{(\mathcal{I}_{i,j})}$ sum to zero, 
{\it i.e.},
$\mathbbm{1}^T \tilde{\bA}^{(\mathcal{I}_{i,j})} =  \textbf{0}^T, 
$
we have:
$$
\mathbbm{1}^T \bR^{(\mathcal{I}_{j})}_{\textbf{2},\textbf{1}} \bpi^{(\mathcal{I}_{j})}_{\textbf{1}} 
- \mathbbm{1}^T \bA^{(\mathcal{I}_{j})}_{\textbf{1},\textbf{2}} \bpi^{(\mathcal{I}_{j})}_{\textbf{2}}
= \mathbbm{1}^T \tilde{\bA}^{(\mathcal{I}_{i,j})} \bpi^{(\mathcal{I}_{j})}_{\textbf{1}} 
=  \textbf{0}^T. 
$$
Therefore, we have: 
$$
\mathbbm{1}^T \bR^{(\mathcal{I}_{j})}_{\textbf{2},\textbf{1}} \bpi^{(\mathcal{I}_{j})}_{\textbf{1}} 
= \mathbbm{1}^T \bA^{(\mathcal{I}_{j})}_{\textbf{1},\textbf{2}} \bpi^{(\mathcal{I}_{j})}_{\textbf{2}}.
$$
As all entries in vector 
$\bR_{\textbf{2},\textbf{1}}^{(\mathcal{I}_j)} 
\bpi^{(\mathcal{I}_{j})}_{\textbf{1}}$ 
and 
$\bA_{\textbf{1},\textbf{2}}^{(\mathcal{I}_j)} 
\bpi^{(\mathcal{I}_{j})}_{\textbf{2}}$ are non-negative, 
we have the following equality of $1$-norms, \textit{i.e.} the summation
of absolute values of vector elements: 
\begin{equation}
\left\| \bR^{(\mathcal{I}_{j})}_{\textbf{2},\textbf{1}} \bpi^{(\mathcal{I}_{j})}_{\textbf{1}} \right\|_1
= \mathbbm{1}^T \bR^{(\mathcal{I}_{j})}_{\textbf{2},\textbf{1}} \bpi^{(\mathcal{I}_{j})}_{\textbf{1}} 
= \mathbbm{1}^T \bA^{(\mathcal{I}_{j})}_{\textbf{1},\textbf{2}} \bpi^{(\mathcal{I}_{j})}_{\textbf{2}}
= \left\| \bA^{(\mathcal{I}_{j})}_{\textbf{1},\textbf{2}} \bpi^{(\mathcal{I}_{j})}_{\textbf{2}} \right\|_1.
\label{eqn:bndryef}
\end{equation}
From Minkowski inequality of vector norm and Eqn.~(\ref{eqn:bndry}), we have:
\begin{equation}
\left\| \tilde{\bA}^{(\mathcal{I}_{i,j})} \bpi^{(\mathcal{I}_{j})}_{\textbf{1}} \right\|_1 
= \left\| \bR^{(\mathcal{I}_{j})}_{\textbf{2},\textbf{1}} \bpi^{(\mathcal{I}_{j})}_{\textbf{1}} - 
\bA^{(\mathcal{I}_{j})}_{\textbf{1},\textbf{2}} \bpi^{(\mathcal{I}_{j})}_{\textbf{2}} \right\|_1
\leq \left\| \bR^{(\mathcal{I}_{j})}_{\textbf{2},\textbf{1}} \bpi^{(\mathcal{I}_{j})}_{\textbf{1}} \right\|_1 + \left\| \bA^{(\mathcal{I}_{j})}_{\textbf{1},\textbf{2}} \bpi^{(\mathcal{I}_{j})}_{\textbf{2}} \right\|_1. 
\label{eqn:bndrynb}
\end{equation}
From Eqn.~(\ref{eqn:bndryef}), we have: 
\begin{equation}
\left\| \tilde{\bA}^{(\mathcal{I}_{i,j})} \bpi^{(\mathcal{I}_{j})}_{\textbf{1}} \right\|_1 \leq 2 \left\| \bR^{(\mathcal{I}_{j})}_{\textbf{2},\textbf{1}} \bpi^{(\mathcal{I}_{j})}_{\textbf{1}} \right\|_1. 
\label{eqn:bndrynb2}
\end{equation}

Now we show that the norm of 
$\left\| \bR_{\textbf{2},\textbf{1}}^{(\mathcal{I}_j)} \bpi^{(\mathcal{I}_{j})}_{\textbf{1}} \right\|_1$ 
converge to zero when the maximum copy number $M$ of the $i$-th MEG goes to infinity. 
In the block tri-diagonal matrix $\tilde{\bA}^{(\mathcal{I}_j)}$, only the boundary block 
$\bA_{M+1,\,M}^{(\mathcal{I}_j)}$ contains nonzero elements in sub-matrix 
$\bA_{\textbf{2},\,\textbf{1}}^{(\mathcal{I}_j)}$, 
and all other blocks in 
$\bA_{\textbf{2},\,\textbf{1}}^{(\mathcal{I}_j)}$ 
contain only zero entries.  
From Cauchy--Schwarz inequality, we have: 
$$
\left\| \bR_{\textbf{2},\textbf{1}}^{(\mathcal{I}_j)} \bpi^{(\mathcal{I}_{j})}_{\textbf{1}} \right\|_1
= 
\left\| \bigl[ \diag(\mathbbm{1}^T \bA_{M+1,M}^{(\mathcal{I}_j)})   \bigr]
        \bigl[ \bpi^{(\mathcal{I}_{j})}_{\textbf{1}}(\mathcal{G}_M) \bigr] \right\|_1  
\leq
\left\| \diag(\mathbbm{1}^T \bA_{M+1,M}^{(\mathcal{I}_j)}) \right\|_1 \cdot \left\| \bpi^{(\mathcal{I}_{j})}_{\textbf{1}}(\mathcal{G}_M) \right\|_1, 
$$
where $\bpi^{(\mathcal{I}_{j})}_{\textbf{1}}(\mathcal{G}_M)$ is the sub-vector corresponding to the 
state partition $\mathcal{G}_M$. 
Furthermore, according to Lemma~\ref{lm:fbs}
and Eqn.~(\ref{eqn:pininf}) after replacing the subscript $i$ in
$\tilde{\pi}_i^{(\infty)}$ with $M$ and taking into consideration of the equivalence of the infinite
space $\Omega^{(\mathcal{I}_{j})}$
and $\Omega^{(\infty)}$ in regard to truncation at $\mathcal{I}_{i}$,
 we have the probability of the boundary block $\mathcal{G}_M$: 
$\left\| \bpi^{(\mathcal{I}_{j})}_{\textbf{1}} (\mathcal{G}_M) \right\|_1 \rightarrow 0$ 
when $M \rightarrow \infty$. 
When synthesis reactions are concentration
independent (zero order reactions) as usually the case~\cite{Nelson2015}, the norm 
$\left\| \diag(\mathbbm{1}^T \bA_{M+1,M}^{(\mathcal{I}_j)}) \right\|_1$ is a constant 
representing the total synthesis rates over states in $\mathcal{G}_M$.  We have: 
$\left\| \bR_{\textbf{2},\textbf{1}}^{(\mathcal{I}_j)} \bpi^{(\mathcal{I}_{j})}_{\textbf{1}} \right\|_1 \rightarrow 0$ 
when $M \rightarrow \infty$. Therefore with Eqn.~(\ref{eqn:bndrynb2}),
we have: 
$$
\lim_{M \rightarrow \infty} \left\| \tilde{\bA}^{(\mathcal{I}_{i,j})} \bpi^{(\mathcal{I}_{j})}_{\textbf{1}} \right\|_1 = 0. 
$$
Hence, 
$$
\lim_{M \rightarrow \infty} \tilde{\bA}^{(\mathcal{I}_{i,j})} \bpi^{(\mathcal{I}_{j})}_{\textbf{1}} = \textbf{0}. 
$$
That is, when the maximum copy number limit of the $i$-th MEG 
is sufficiently large, both $\bpi^{(\mathcal{I}_{i,j})}$ and 
$\bpi^{(\mathcal{I}_{j})}_{\textbf{1}}$ are the steady state 
solutions of $\tilde{\bA}^{(\mathcal{I}_{i,j})} \by = 0$. 
According to Perron--Frobenius theorem for the transition rate 
matrix of continuous-time Markov chains~\cite{Meyer2000}, the dCME governed by
$\tilde{\bA}^{(\mathcal{I}_{i,j})}$ has a globally unique stationary
distribution.  In addition, 
by construction of the matrix, 
via enumeration of the state space,
 matrix $\tilde{\bA}^{(\mathcal{I}_{i,j})}$ is
irreducible, as all microstates in the state space can be 
reached from the initial state.  Therefore, 
the matrix $\tilde{\bA}^{(\mathcal{I}_{i,j})}$ 
has only one zero eigenvalue~\cite{Meyer2000}, both $\bpi^{(\mathcal{I}_{i,j})}$ and 
$\bpi^{(\mathcal{I}_{j})}_{\textbf{1}}$ are eigenvectors corresponding
to the eigenvalue $0$.  Therefore, we have the relationship 
$\bpi^{(\mathcal{I}_{i,j})} = c \bpi^{(\mathcal{I}_{j})}_{\textbf{1}}$, 
where $c$ is an arbitrary real number. 
As both vectors are non-negative, and 
$\mathbbm{1}^T \bpi^{(\mathcal{I}_{j})}_{\textbf{1}} \leq 1 = \mathbbm{1}^T \bpi^{(\mathcal{I}_{i,j})}$, 
there must exist an $\epsilon = 1 - \mathbbm{1}^T \bpi^{(\mathcal{I}_{j})}_{\textbf{1}} \geq 0$, 
such that $\bpi^{(\mathcal{I}_{i,j})} = (1 + \epsilon) \bpi^{(\mathcal{I}_{j})}_{\textbf{1}}$. 
According to Lemma~\ref{lm:fbs}, $\mathbbm{1}^T \bpi^{(\mathcal{I}_{j})}_{\textbf{1}} \rightarrow 1$, 
when the maximum copy number limit of the $i$-th MEG goes to infinity. 
Therefore we have $\epsilon \rightarrow 0$ when $M \rightarrow \infty$. 
Therefore, we have shown both $\bpi^{(\mathcal{I}_{i,j})} \geq \bpi^{(\mathcal{I}_{j})}_{\textbf{1}}$ 
and $\bpi^{(\mathcal{I}_{i,j})} \rightarrow \bpi^{(\mathcal{I}_{j})}_{\textbf{1}}$ component-wise, 
when the maximum copy number limit of the $i$-th MEG goes to infinity.

\end{proof}

\bibliographystyle{unsrt}      
\bibliography{dcme}

\begin{thebibliography}{10}

\bibitem{Stewart2012}
Jacob Stewart-Ornstein and Hana El-Samad.
\newblock Stochastic modeling of cellular networks.
\newblock {\em Computational Methods in Cell Biology}, 110:111, 2012.

\bibitem{Qian2012}
Hong Qian.
\newblock Cooperativity in cellular biochemical processes: noise-enhanced
  sensitivity, fluctuating enzyme, bistability with nonlinear feedback, and
  other mechanisms for sigmoidal responses.
\newblock {\em Annual Review of Biophysics}, 41:179--204, 2012.

\bibitem{McAdams1999}
H.H. McAdams and A.~Arkin.
\newblock {It's a noisy business! Genetic regulation at the nanomolar scale}.
\newblock {\em Trends in Genetics}, 15(2):65--69, 1999.

\bibitem{Wilkinson2009}
Darren~J Wilkinson.
\newblock Stochastic modelling for quantitative description of heterogeneous
  biological systems.
\newblock {\em Nature Reviews Genetics}, 10(2):122--133, 2009.

\bibitem{Cao2010}
Youfang Cao, Hsiao-Mei Lu, and Jie Liang.
\newblock {Probability landscape of heritable and robust epigenetic state of
  lysogeny in phage lambda}.
\newblock {\em Proceedings of the National Academy of Sciences of the United
  States of America}, 107(43):18445--18450, 2010.

\bibitem{Gillespie_JPC77}
D.~T. Gillespie.
\newblock Exact stochastic simulation of coupled chemical reactions.
\newblock {\em Journal of Physical Chemistry}, 81:2340--2361, 1977.

\bibitem{Gillespie-PhysicaA-1992}
Daniel~T. Gillespie.
\newblock A rigorous derivation of the chemical master equation.
\newblock {\em Physica A}, 188:404--425, 1992.

\bibitem{vanKampen2007}
N.G. Van~Kampen.
\newblock {\em Stochastic processes in physics and chemistry, 3rd Edition}.
\newblock Elsevier Science and Technology books, 2007.

\bibitem{Beard2008}
D.A. Beard and H.~Qian.
\newblock {\em {Chemical biophysics: quantitative analysis of cellular
  systems}}.
\newblock Cambridge Univ Pr, 2008.

\bibitem{Gillespie2009-jcp}
Daniel~T. Gillespie.
\newblock A diffusional bimolecular propensity function.
\newblock {\em The Journal of Chemical Physics}, 131(16):164109, 2009.

\bibitem{Darvey-1966-jcp}
I.G. Darvey, B.W. Ninham, and P.J. Staff.
\newblock Stochastic models for second order chemical reaction kinetics. the
  equilibrium state.
\newblock {\em The Journal of Chemical Physics}, 45:2145--2155, 1966.

\bibitem{McQuarrie-1967-JAppProb}
D.A. McQuarrie.
\newblock Stochastic approach to chemical kinetics.
\newblock {\em Journal of Applied Probability}, 4:413--478, 1967.

\bibitem{Taylor1998}
H.M. Taylor and S.~Karlin.
\newblock {\em {An Introduction to Stochastic Modeling, 3rd Ed.}}
\newblock Academic Press, 1998.

\bibitem{Laurenzi-2000-jcp}
I.J. Laurenzi.
\newblock An analytical solution of the stochastic master equation for
  reversible bimolecular reaction kinetics.
\newblock {\em The Journal of Chemical Physics}, 113:3315--3322, 2000.

\bibitem{Vellela2007}
Melissa Vellela and Hong Qian.
\newblock A quasistationary analysis of a stochastic chemical reaction:
  {K}eizer’s paradox.
\newblock {\em Bulletin of Mathematical Biology}, 69(5):1727--1746, 2007.

\bibitem{Gillespie_JCP2000}
Daniel.~T. Gillespie.
\newblock The chemical langevin equation.
\newblock {\em The Journal of Chemical Physics}, 113:297--306, 2000.

\bibitem{VanKampen-1961}
N.~G. Van~Kampen.
\newblock A power series expansion of the master equation.
\newblock {\em Canadian Journal of Physics}, 39(4):551--567, 1961.

\bibitem{Gillespie2002}
Daniel~T. Gillespie.
\newblock The chemical {L}angevin and {F}okker−{P}lanck equations for the
  reversible isomerization reaction†.
\newblock {\em The Journal of Physical Chemistry A}, 106(20):5063--5071, 2002.

\bibitem{Haseltine2002}
Eric~L. Haseltine and James~B. Rawlings.
\newblock Approximate simulation of coupled fast and slow reactions for
  stochastic chemical kinetics.
\newblock {\em The Journal of Chemical Physics}, 117(15):6959--6969, 2002.

\bibitem{Gardiner2004-book}
C.~W. Gardiner.
\newblock {\em Handbook of Stochastic Methods for Physics, Chemistry and the
  Natural Sciences}.
\newblock Springer, New York, 2004.

\bibitem{Grima-2011-jcp}
R.~Grima, P.~Thomas, and A.~V. Straube.
\newblock {{H}ow accurate are the nonlinear chemical {F}okker-{P}lanck and
  chemical {L}angevin equations?}
\newblock {\em The Journal of Chemical Physics}, 135(8):084103, Aug 2011.

\bibitem{Grima-2013-bcmgenomics}
P.~Thomas, H.~Matuschek, and R.~Grima.
\newblock {{H}ow reliable is the linear noise approximation of gene regulatory
  networks?}
\newblock {\em BMC Genomics}, 14 Suppl 4:S5, 2013.

\bibitem{Munsky2006}
Brian Munsky and Mustafa Khammash.
\newblock The finite state projection algorithm for the solution of the
  chemical master equation.
\newblock {\em The Journal of Chemical Physics}, 124(4):044104, 2006.

\bibitem{CaoBMCSB08}
Youfang Cao and Jie Liang.
\newblock {Optimal enumeration of state space of finitely buffered stochastic
  molecular networks and exact computation of steady state landscape
  probability}.
\newblock {\em BMC Systems Biology}, 2(1):30, 2008.

\bibitem{MacNamara2008a}
Shev MacNamara, Alberto~M Bersani, Kevin Burrage, and Roger~B Sidje.
\newblock Stochastic chemical kinetics and the total quasi-steady-state
  assumption: application to the stochastic simulation algorithm and chemical
  master equation.
\newblock {\em The Journal of chemical physics}, 129(9):095105, 2008.

\bibitem{MacNamara2008b}
Shev MacNamara, Kevin Burrage, and Roger~B Sidje.
\newblock Multiscale modeling of chemical kinetics via the master equation.
\newblock {\em Multiscale Modeling \& Simulation}, 6(4):1146--1168, 2008.

\bibitem{Wolf2010}
Verena Wolf, Rushil Goel, Maria Mateescu, and Thomas Henzinger.
\newblock Solving the chemical master equation using sliding windows.
\newblock {\em BMC Systems Biology}, 4(1):42, 2010.

\bibitem{Jahnke2011}
Tobias Jahnke.
\newblock On reduced models for the chemical master equation.
\newblock {\em Multiscale Modeling \& Simulation}, 9(4):1646--1676, 2011.

\bibitem{Sidje1998}
Roger~B Sidje.
\newblock Expokit: a software package for computing matrix exponentials.
\newblock {\em ACM Transactions on Mathematical Software (TOMS)},
  24(1):130--156, 1998.

\bibitem{Munsky2007}
Brian Munsky and Mustafa Khammash.
\newblock A multiple time interval finite state projection algorithm for the
  solution to the chemical master equation.
\newblock {\em Journal of Computational Physics}, 226(1):818 -- 835, 2007.

\bibitem{Tian2006}
Jianjun~Paul Tian and D~Kannan.
\newblock Lumpability and commutativity of {M}arkov processes.
\newblock {\em Stochastic analysis and Applications}, 24(3):685--702, 2006.

\bibitem{Truffet1997}
Laurent Truffet.
\newblock Near complete decomposability: bounding the error by a stochastic
  comparison method.
\newblock {\em Advances in Applied Probability}, pages 830--855, 1997.

\bibitem{Buchholz1994}
Peter Buchholz.
\newblock Exact and ordinary lumpability in finite {M}arkov chains.
\newblock {\em Journal of Applied Probability}, pages 59--75, 1994.

\bibitem{Stewart1994}
W.J. Stewart.
\newblock {\em {Introduction to the numerical solution of Markov chains}}.
\newblock Princeton University Press NJ, 1994.

\bibitem{Vantilborgh1985}
Hendrik Vantilborgh.
\newblock Aggregation with an error of o$(\epsilon^2)$.
\newblock {\em Journal of the ACM (JACM)}, 32(1):162--190, 1985.

\bibitem{Kemeny1976}
John~G Kemeny and James~Laurie Snell.
\newblock {\em Finite {M}arkov chains}, volume 210.
\newblock Springer-Verlag New York, 1976.

\bibitem{Irle2003}
A~Irle.
\newblock Stochastic ordering for continuous-time processes.
\newblock {\em Journal of Applied Probability}, pages 361--375, 2003.

\bibitem{Daigle2011}
B.J. Daigle, M.K. Roh, D.T. Gillespie, and L.R. Petzold.
\newblock Automated estimation of rare event probabilities in biochemical
  systems.
\newblock {\em The Journal of Chemical Physics}, 134:044110, 2011.

\bibitem{Roh2011}
M.K. Roh, B.J. Daigle, D.T. Gillespie, and L.R. Petzold.
\newblock State-dependent doubly weighted stochastic simulation algorithm for
  automatic characterization of stochastic biochemical rare events.
\newblock {\em Journal of Chemical Physics}, 135(23):234108, 2011.

\bibitem{Cao2013JCP}
Youfang Cao and Jie Liang.
\newblock Adaptively biased sequential importance sampling for rare events in
  reaction networks with comparison to exact solutions from finite buffer
  d{CME} method.
\newblock {\em The Journal of Chemical Physics}, 139(2):025101, 2013.

\bibitem{Taniguchi2010}
Yuichi Taniguchi, Paul~J Choi, Gene-Wei Li, Huiyi Chen, Mohan Babu, Jeremy
  Hearn, Andrew Emili, and X~Sunney Xie.
\newblock Quantifying e. coli proteome and transcriptome with single-molecule
  sensitivity in single cells.
\newblock {\em Science}, 329(5991):533--538, 2010.

\bibitem{Gardner2000}
Timothy~S Gardner, Charles~R Cantor, and James~J Collins.
\newblock Construction of a genetic toggle switch in {E}scherichia coli.
\newblock {\em Nature}, 403(6767):339--342, 2000.

\bibitem{Kepler2001}
Thomas~B Kepler and Timothy~C Elston.
\newblock Stochasticity in transcriptional regulation: origins, consequences,
  and mathematical representations.
\newblock {\em Biophysical Journal}, 81(6):3116--3136, 2001.

\bibitem{Kim2007}
Keun-Young Kim and Jin Wang.
\newblock Potential energy landscape and robustness of a gene regulatory
  network: toggle switch.
\newblock {\em PLoS Computational Biology}, 3(3):e60, 2007.

\bibitem{Schultz_JCP07}
D.~Schultz, J.~N. Onuchic, and P.~G. Wolynes.
\newblock Understanding stochastic simulations of the smallest genetic
  networks.
\newblock {\em The Journal of Chemical Physics}, 126(24):245102, 2007.

\bibitem{Munsky2008}
Brian Munsky and Mustafa Khammash.
\newblock The finite state projection approach for the analysis of stochastic
  noise in gene networks.
\newblock {\em Automatic Control, IEEE Transactions on}, 53(Special
  Issue):201--214, 2008.

\bibitem{Deuflhard2008}
Peter Deuflhard, Wilhelm Huisinga, T~Jahnke, and Michael Wulkow.
\newblock Adaptive discrete {G}alerkin methods applied to the chemical master
  equation.
\newblock {\em SIAM Journal on Scientific Computing}, 30(6):2990--3011, 2008.

\bibitem{Sjoberg2009}
Paul Sj{\"o}berg, Per L{\"o}tstedt, and Johan Elf.
\newblock {F}okker--{P}lanck approximation of the master equation in molecular
  biology.
\newblock {\em Computing and Visualization in Science}, 12(1):37--50, 2009.

\bibitem{Kazeev2014}
Vladimir Kazeev, Mustafa Khammash, Michael Nip, and Christoph Schwab.
\newblock Direct solution of the chemical master equation using quantized
  tensor trains.
\newblock {\em PLoS Computational Biology}, 10(3):e1003359, 03 2014.

\bibitem{Arkin1998}
Adam Arkin, John Ross, and Harley~H. McAdams.
\newblock {Stochastic kinetic analysis of developmental pathway bifurcation in
  phage $\lambda$-infected {E}scherichia coli cells}.
\newblock {\em Genetics}, 149(4):1633--1648, 1998.

\bibitem{Aurell2002PRE}
Erik Aurell, Stanley Brown, Johan Johanson, and Kim Sneppen.
\newblock {Stability puzzles in phage $\lambda$}.
\newblock {\em Physical Review E}, 65(5):051914, 2002.

\bibitem{Aurell2002PRL}
Erik Aurell and Kim Sneppen.
\newblock {Epigenetics as a first exit problem}.
\newblock {\em Physical Review Letters}, 88(4):048101, 2002.

\bibitem{ZhuJBCB2004}
X.-M. Zhu, L.~Yin, L.~Hood, and P.~Ao.
\newblock Robustness, stability and efficiency of phage lambda genetic switch:
  dynamical structure analysis.
\newblock {\em Journal of Bioinformatics and Computational Biology},
  2:785--817, 2004.

\bibitem{Zhu2004}
X.-M. Zhu, L.~Yin, L.~Hood, and P.~Ao.
\newblock {Calculating biological behaviors of epigenetic states in the phage
  $\lambda$ life cycle}.
\newblock {\em Functional \& Integrative Genomics}, 4(3):188--195, 2004.

\bibitem{Liao2015}
Shuohao Liao, T~Vejchodsky, and Radek Erban.
\newblock Tensor methods for parameter estimation and bifurcation analysis of
  stochastic reaction networks.
\newblock {\em Journal of the Royal Society Interface}, 12(108):20150233, 2015.

\bibitem{Verstraete2006}
F.~Verstraete and J.~I. Cirac.
\newblock Matrix product states represent ground states faithfully.
\newblock {\em Phys. Rev. B}, 73:094423, Mar 2006.

\bibitem{Nelson2015}
Philip Nelson.
\newblock {\em Physical Models of Living Systems}.
\newblock Macmillan, 2015.

\bibitem{Meyer2000}
Carl~D Meyer.
\newblock {\em Matrix analysis and applied linear algebra}.
\newblock SIAM, 2000.

\bibitem{Li_PNASUSA97}
M.~Li, W.~McClure, and M.~Susskind.
\newblock Changing the mechanism of transcriptional activation by phage lambda
  repressor.
\newblock {\em Proceedings of the National Academy of Sciences of the United
  States of America}, 94(8):3691--3696, 1997.

\bibitem{Hawley_PNASUSA80}
D.~Hawley and W.~McClure.
\newblock In vitro comparison of initiation properties of bacteriophage lambda
  wild-type {PR} and x3 mutant promoters.
\newblock {\em Proceedings of the National Academy of Sciences of the United
  States of America}, 77(11):6381--6385, 1980.

\bibitem{Hawley_JMB82}
D.~Hawley and W.~McClure.
\newblock Mechanism of activation of transcription initiation from the lambda
  {PRM} promoter.
\newblock {\em Journal of Molecular Biology}, 157(3):493--525, 1982.

\bibitem{Shea1985}
Madeline~A. Shea and Gary~K. Ackers.
\newblock {The $OR$ control system of bacteriophage lambda a physical-chemical
  model for gene regulation}.
\newblock {\em Journal of Molecular Biology}, 181(2):211--230, 1985.

\bibitem{Kuttler2006}
C\'{e}line Kuttler and Joachim Niehren.
\newblock {Gene Regulation in the Pi Calculus: Simulating Cooperativity at the
  Lambda Switch}.
\newblock {\em Transactions on Computational Systems Biology VII}, 4230:24--55,
  2006.

\end{thebibliography}

\newpage
\clearpage

\section*{Figure Legends}

\begin{figure}[ht]
\centering{\includegraphics[scale=0.8]{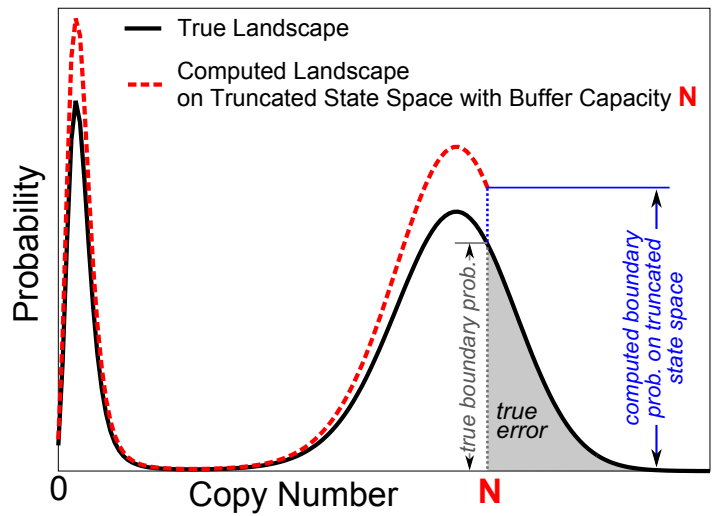}}
\caption{ An illustration for the boundary probabilities. The solid black 
line represents the true probability landscape on the exact 
state space. The dashed red line represents the probability landscape 
computed from the truncated state space with buffer capacity $N$. 
The gray shaded area represents the true error due to state space truncation 
with buffer capacity $N$. The probability of copy number $N$ on the true 
landscape is the true boundary probability, and the probability of $N$ 
on the computed landscape is the boundary probability on the 
truncated state space. In this study, we show that the computed 
boundary probability on the truncated state space can be used to 
bound the true error from the above. }
\label{fig:BN}
\end{figure}

\begin{figure}[ht]
\centering{\includegraphics[scale=0.5]{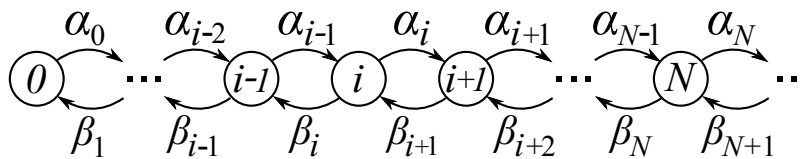}}
\caption{The birth-death system associated with the aggregated rate
  matrix $\bB$. Each circle represents an aggregated state consisting of
  all microstates with the same copy number of elementary molecules in the MEG. 
	These aggregated states are connected by aggregated birth and
  death reactions, with apparent synthesis rates $\alpha_i$ and
  degradation rates $\beta_{i+1}$ (see Lemma~\ref{lm:rma}).  }
\label{fig:bfbd}
\end{figure}

\begin{figure}[ht]
\centering{
\includegraphics[scale=0.4]{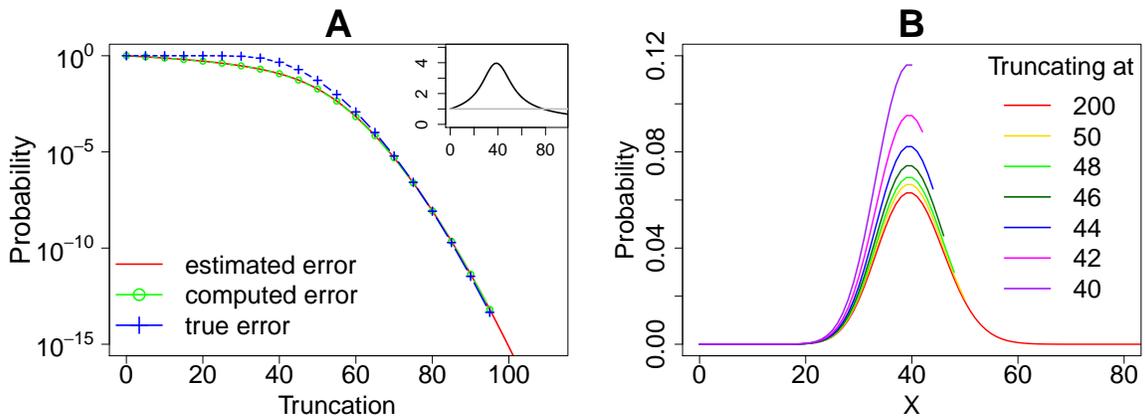}
}
\caption{Error quantification and comparisons for the birth-death model. 
  (A): The {\it a priori}\/ estimated error (red solid line), the
  computed error (green line and circles), and the true error
  (blue line and crosses) of the steady state probability
  landscape.  
The inset shows the ratio of the true errors to the computed errors at different 
sizes of the MEG, and the grey straight line marks the ratio one. 
The computed errors are larger than the true errors when the black line 
is below the grey straight line. 
(B): The steady state probability landscapes of $X$ obtained  
	with different truncations of net molecular number in the MEG. 
	Note that probability distributions end at $X$ where truncation occurs. 
	The probabilities in the landscapes are inflated when truncating the 
	state space at smaller net molecular numbers of the MEG. }
\label{fig:bd1}
\end{figure}

\begin{figure}[ht]
\centering{
\includegraphics[scale=0.5]{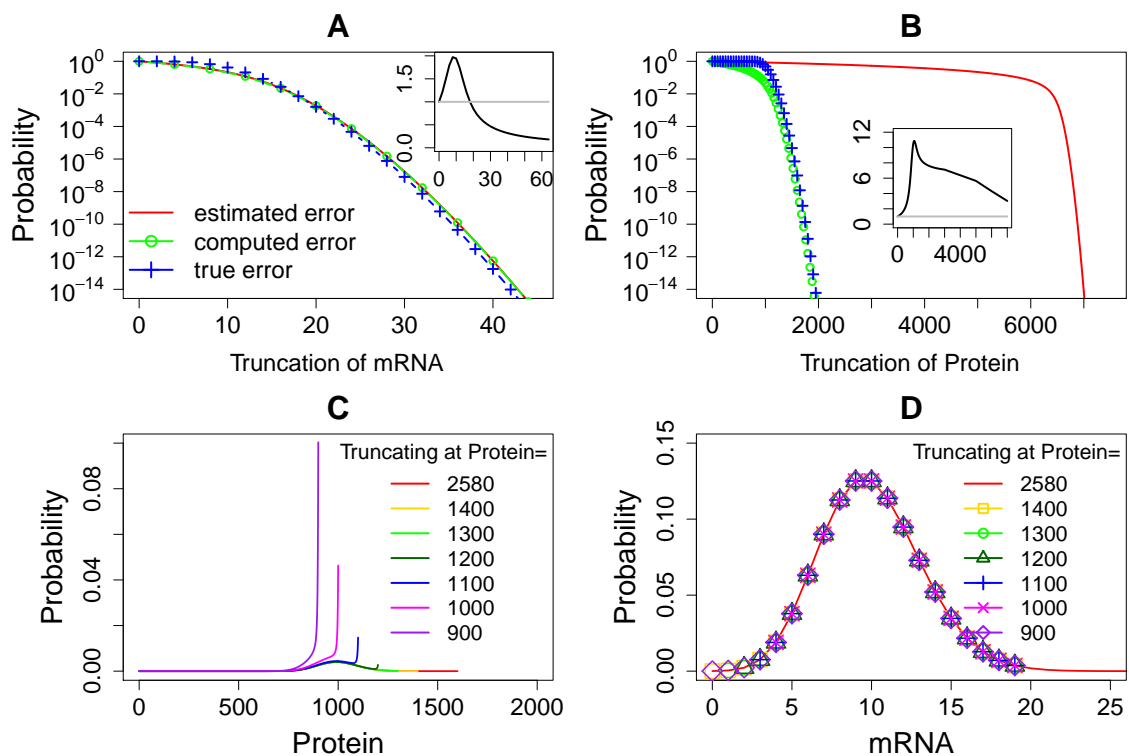}
}
\caption{Error quantification and comparisons for the single gene expression model. 
  (A) and (B): The {\it a priori}\/ estimated error (red solid lines), the
  computed error (green lines and circles), and the true error
  (blue lines and crosses) of the steady state probability
  landscapes of $mRNA$ and $Protein$ at different sizes of truncations.  
The insets in (A) and (B) show the ratio of the true errors to the computed errors at different 
sizes of the MEG, and the grey straight line marks the ratio one. 
The computed errors are larger than the true errors when the black line 
is below the grey straight line. 
	(C): The steady state probability landscapes of $Protein$ 
	solved using different truncations of net molecular number in the MEG$_2$. 
	Note that probability distributions end at where truncation occurs. 
	The probabilities in the landscapes are significantly inflated when truncating the 
	state space at smaller net molecular numbers of the corresponding MEG. 
        (D): The steady state probability landscapes of $mRNA$ 
	solved using different truncations of net molecular number in the MEG$_2$. 
	In this cases, the probabilities in the landscapes are not affected by the 
        truncation of the opposite MEG. } 
\label{fig:sge1}
\end{figure}

\begin{figure}[ht]
\centering{
\includegraphics[scale=0.5]{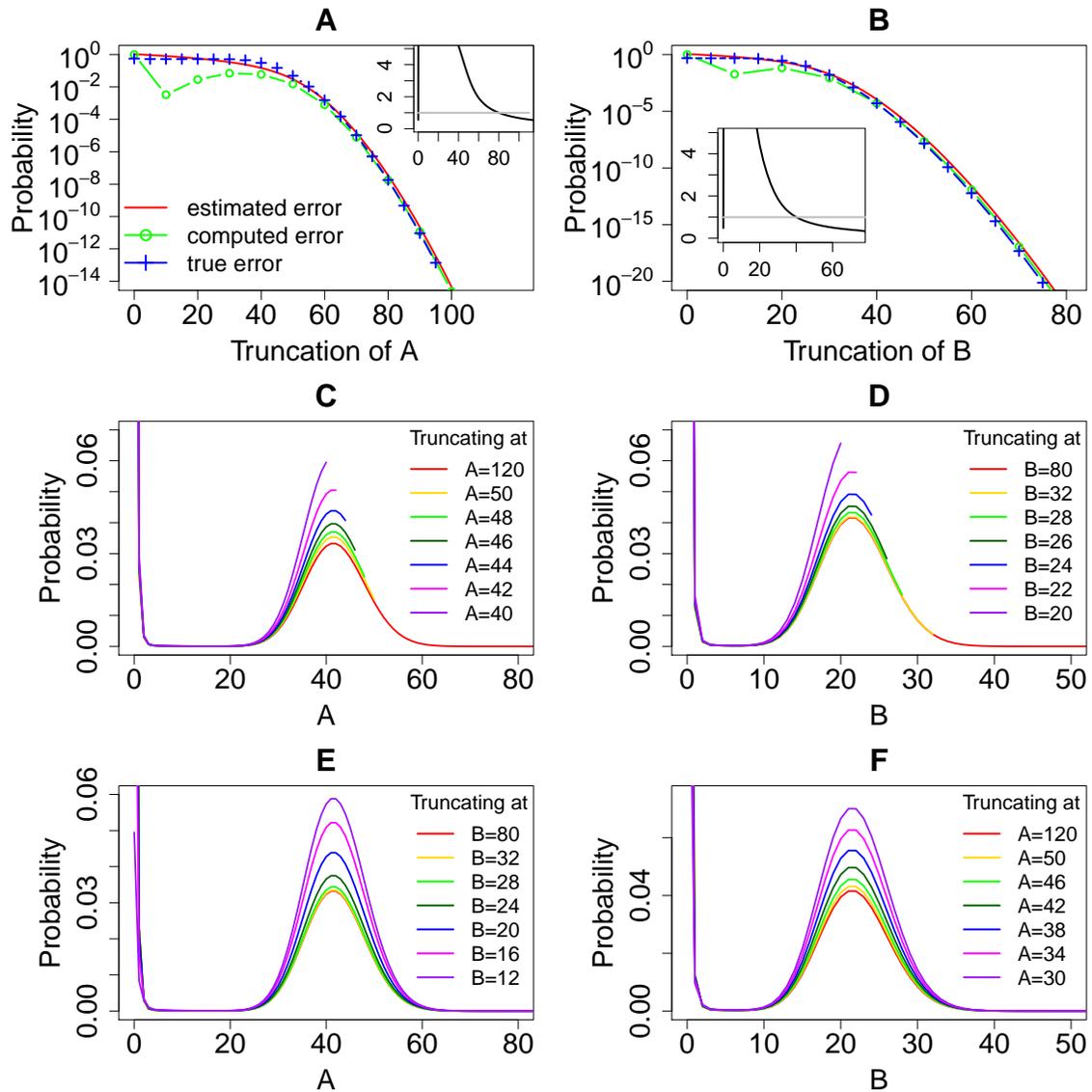} 
}
\caption{Error quantification and comparisons for the genetic toggle switch model. 
  (A) and (B): The {\it a priori}\/ estimated error (red solid lines), the
  computed error (green lines and circles), and the true error
  (blue lines and crosses) of the steady state probability
  landscapes of $A$ and $B$ at different sizes of truncations.  
The insets in (A) and (B) show the ratio of the true errors to the computed errors at different 
sizes of the MEG, and the grey straight line marks the ratio one. 
The computed errors are larger than the true errors when the black line 
is below the grey straight line. 
	(C) and (D): The steady state probability landscapes of $A$ and $B$ 
	solved using different truncations of net molecular number in the MEG$_1$ and MEG$_2$, respectively. 
	Note that probability distributions end at where truncation occurs. 
	The probabilities in the landscapes are significantly inflated when truncating the 
	state space at smaller net molecular numbers of the corresponding MEG. 
  (E) and (F): The steady state probability landscapes of $A$ and $B$ 
	solved using different truncations of net molecular number in the MEG$_2$ and MEG$_1$, respectively. 
	The probabilities in the landscapes are also significantly inflated when truncating the 
	state space at smaller net molecular numbers of the opposite MEG. }
\label{fig:tg1}
\end{figure}

\begin{figure}[ht]
\centering{
\includegraphics[scale=0.5]{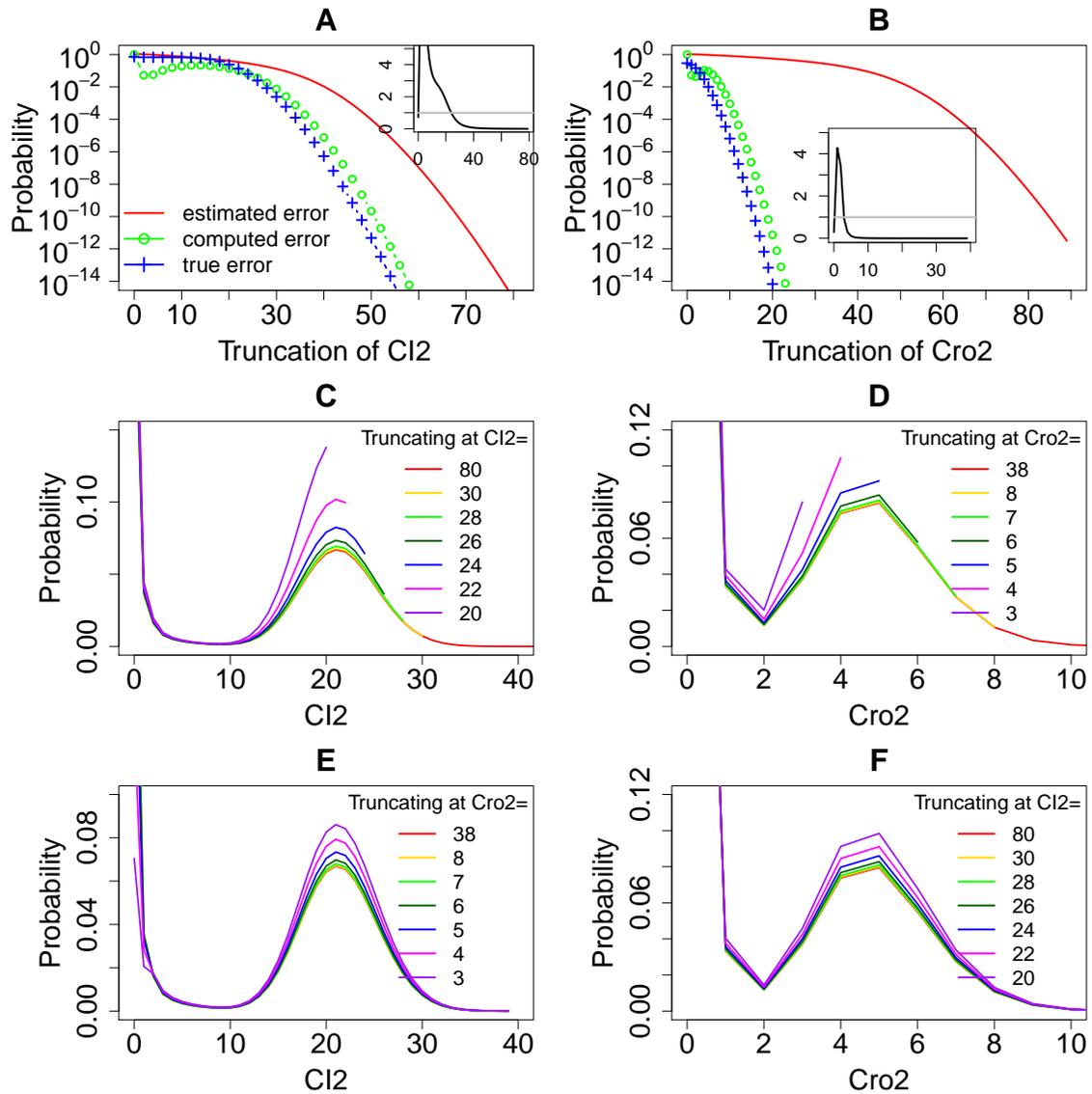} 
}
\caption{Error quantification and comparisons for the phage lambda bistable epigenetic switch model. 
  (A) and (B): The {\it a priori}\/ estimated error (red solid lines), the
  computed error (green lines and circles), and the true error
  (blue lines and crosses) of the steady state probability
  landscapes of $CI$ and $Cro$ dimers at different sizes of truncations.  
The insets in (A) and (B) show the ratio of the true errors to the computed errors at different 
sizes of the MEG, and the grey straight line marks the ratio one. 
The computed errors are larger than the true errors when the black line 
is below the grey straight line. 
	(C) and (D): The steady state probability landscapes of $CI$ and $Cro$ dimers 
	solved using different truncations of net molecular number in the MEG$_1$ and MEG$_2$, respectively. 
	Note that probability distributions end at where truncation occurs. 
	The probabilities in the landscapes are significantly inflated when truncating the 
	state space at smaller net molecular numbers of the corresponding MEG. 
  (E) and (F): The steady state probability landscapes of $CI$ and $Cro$ dimers 
	solved using different truncations of net molecular number in the MEG$_2$ and MEG$_1$, respectively. 
	The probabilities in the landscapes are also significantly inflated when truncating the 
	state space at smaller net molecular numbers of the opposite MEG. }
\label{fig:ph1}
\end{figure}

\begin{table}[!ht]
\caption{Reaction scheme and rate constants in phage lambda epigenetic switch mode. 
We use $CORn$ denotes $Cro_2$ bound operator site $ORn$, $RORn$ denotes $CI_2$
bound $ORn$, where $n$ can be 1, 2, and 3.  Note that molecular species enclosed in
parenthesis are those whose presence is required for the specific
reactions to occur, but their copy numbers do not influence the
transition rates between microstates. } 
\label{tab:phage1}
\begin{scriptsize}
\begin{center}
\begin{tabular}{|ll|}
  \hline
  Reactions & Rate constants \\
  \hline
  \multicolumn{2}{|l|}{Synthesis reactions \cite{Arkin1998,Li_PNASUSA97,Hawley_PNASUSA80,Hawley_JMB82}} \\
  \hline
  $\emptyset + (OR3 + OR2) \rightarrow CI_{2} + (OR3 + OR2)$ & $k_{sCI_{2}}=0.0069/s$ \\
  $\emptyset + (OR3 + COR2) \rightarrow CI_{2} + (OR3 + COR2)$ &  $k_{sCI_{2}}=0.0069/s$\\
  $\emptyset + (OR3 + ROR2) \rightarrow CI_{2} + (OR3 + ROR2)$ &  $k_{s1CI_{2}}=0.069/s$\\
  $\emptyset + (OR1 + OR2) \rightarrow Cro_{2} + (OR1 + OR2)$ &  $k_{sCro_{2}}=0.0929/s$ \\
  \hline
  \multicolumn{2}{|l|}{Degradation reactions \cite{Shea1985,Arkin1998}} \\
  \hline
  $CI_{2} \rightarrow \emptyset$ & $k_{dCI_{2}}=0.0026/s$ \\ 
  $Cro_{2} \rightarrow \emptyset$ & $k_{dCro_{2}}=0.0025/s$\\
  \hline
  \multicolumn{2}{|l|}{Association rate of binding reactions~\cite{Kuttler2006}} \\
  \hline
  $CI_{2} + OR1 \rightarrow ROR1$ & $k_{bOR1CI_{2}}=0.021/s$ \\
  $CI_{2} + OR2 \rightarrow ROR2$ & $k_{bOR2CI_{2}}=0.021/s$ \\
  $CI_{2} + OR3 \rightarrow ROR3$ & $k_{bOR3CI_{2}}=0.021/s$ \\
  $Cro_{2} + OR1 \rightarrow COR1$ & $k_{bOR1Cro_{2}}=0.021/s$ \\
  $Cro_{2} + OR2 \rightarrow COR2$ & $k_{bOR2Cro_{2}}=0.021/s$ \\
  $Cro_{2} + OR3 \rightarrow COR3$ & $k_{bOR3Cro_{2}}=0.021/s$ \\
  \hline
  \multicolumn{2}{|l|}{Dissociation reactions - $CI_{2}$ dissociation from OR1} \\
  \hline
  $ROR1 + (OR2) \rightarrow CI_{2} + OR1 + (OR2)$ & $0.00898/s$ \\
  $ROR1 + (ROR2 + OR3) \rightarrow CI_{2} + OR1 + (ROR2 + OR3)$ & $0.00011/s$ \\
  $ROR1 + (ROR2 + ROR3) \rightarrow CI_{2} + OR1 + (ROR2 + ROR3)$ & $0.01242/s$ \\
  $ROR1 + (ROR2 + COR3) \rightarrow CI_{2} + OR1 + (ROR2 + COR3)$ & $0.00011/s$ \\
  $ROR1 + (COR2) \rightarrow CI_{2} + OR1 + (COR2)$ & $0.00898/s$ \\
  \hline
  \multicolumn{2}{|l|}{Dissociation reactions - $CI_{2}$ dissociation from OR2} \\
  \hline
  $ROR2 + (OR1 + OR3) \rightarrow CI_{2} + OR2 + (OR1 + OR3)$ & $0.2297/s$ \\
  $ROR2 + (ROR1 + OR3) \rightarrow CI_{2} + OR2 + (ROR1 + OR3)$ & $0.0029/s$ \\
  $ROR2 + (OR1 + ROR3) \rightarrow CI_{2} + OR2 + (OR1 + ROR3)$ & $0.0021/s$ \\
  $ROR2 + (ROR1 + ROR3) \rightarrow CI_{2} + OR2 + (ROR1 + ROR3)$ & $0.0029/s$ \\
  $ROR2 + (COR1 + OR3) \rightarrow CI_{2} + OR2 + (COR1 + OR3)$ & $0.2297/s$ \\
  $ROR2 + (OR1 + COR3) \rightarrow CI_{2} + OR2 + (OR1 + COR3)$ & $0.2297/s$ \\
  $ROR2 + (COR1 + COR3) \rightarrow CI_{2} + OR2 + (COR1 + COR3)$ & $0.2297/s$ \\
  $ROR2 + (ROR1 + COR3) \rightarrow CI_{2} + OR2 + (ROR1 + COR3)$ & $0.0029/s$ \\
  $ROR2 + (COR1 + ROR3) \rightarrow CI_{2} + OR2 + (COR1 + ROR3)$ & $0.0021/s$ \\
  \hline
  \multicolumn{2}{|l|}{Dissociation reactions - CI dissociation from OR3} \\
  \hline
  $ROR3 + (OR2) \rightarrow CI_{2} + OR3 + (OR2)$ & $1.13/s$ \\
  $ROR3 + (ROR2 + OR1) \rightarrow CI_{2} + OR3 + (ROR2 + OR1)$ & $0.0106/s$ \\
  $ROR3 + (ROR2 + ROR1) \rightarrow CI_{2} + OR3 + (ROR2 + ROR1)$ & $0.0106/s$ \\
  $ROR3 + (ROR2 + COR1) \rightarrow CI_{2} + OR3 + (ROR2 + COR1)$ & $0.0106/s$ \\
  $ROR3 + (COR2) \rightarrow CI_{2} + OR3 + (COR2)$ & $1.13/s$ \\
  \hline
  \multicolumn{2}{|l|}{Dissociation reactions - Cro dissociation from OR1} \\
  \hline
  $COR1 + (OR2) \rightarrow Cro_{2} + OR1 + (OR2)$ & $0.0202/s$ \\
  $COR1 + (ROR2) \rightarrow Cro_{2} + OR1 + (ROR2)$ & $0.0202/s$ \\
  $COR1 + (COR2 + OR3) \rightarrow Cro_{2} + OR1 + (COR2 + OR3)$ & $0.0040/s$ \\
  $COR1 + (COR2 + ROR3) \rightarrow Cro_{2} + OR1 + (COR2 + ROR3)$ & $0.0040/s$ \\
  $COR1 + (COR2 + COR3) \rightarrow Cro_{2} + OR1 + (COR2 + COR3)$ & $0.0040/s$ \\
  \hline
  \multicolumn{2}{|l|}{Dissociation reactions - Cro dissociation from OR2} \\
  \hline
  $COR2 + (OR1 + OR3) \rightarrow Cro_{2} + OR2 + (OR1 + OR3)$ & $0.1413/s$ \\
  $COR2 + (ROR1 + OR3) \rightarrow Cro_{2} + OR2 + (ROR1 + OR3)$ & $0.1413/s$ \\
  $COR2 + (OR1 + ROR3) \rightarrow Cro_{2} + OR2 + (OR1 + ROR3)$ & $0.1413/s$ \\
  $COR2 + (ROR1 + ROR3) \rightarrow Cro_{2} + OR2 + (ROR1 + ROR3)$ & $0.1413/s$ \\
  $COR2 + (COR1 + OR3) \rightarrow Cro_{2} + OR2 + (COR1 + OR3)$ & $0.0279/s$ \\
  $COR2 + (OR1 + COR3) \rightarrow Cro_{2} + OR2 + (OR1 + COR3)$ & $0.053/s$ \\
  $COR2 + (COR1 + COR3) \rightarrow Cro_{2} + OR2 + (COR1 + COR3)$ & $0.0328/s$ \\
  $COR2 + (ROR1 + COR3) \rightarrow Cro_{2} + OR2 + (ROR1 + COR3)$ & $0.053/s$ \\
  $COR2 + (COR1 + ROR3) \rightarrow Cro_{2} + OR2 + (COR1 + ROR3)$ & $0.0279/s$ \\
  \hline
  \multicolumn{2}{|l|}{Dissociation reactions - Cro dissociation from OR3} \\
  \hline
  $COR3 + (OR2) \rightarrow Cro_{2} + OR3 + (OR2)$ & $0.0022/s$ \\
  $COR3 + (ROR2) \rightarrow Cro_{2} + OR3 + (ROR2)$ & $0.0022/s$ \\
  $COR3 + (COR2 + OR1) \rightarrow Cro_{2} + OR3 + (COR2 + OR1)$ & $0.0008/s$ \\
  $COR3 + (COR2 + ROR1) \rightarrow Cro_{2} + OR3 + (COR2 + ROR1)$ & $0.0008/s$ \\
  $COR3 + (COR2 + COR1) \rightarrow Cro_{2} + OR3 + (COR2 + COR1)$ & $0.003/s$ \\
  \hline
\end{tabular}
\end{center}
\end{scriptsize}
\end{table}

\end{document}